\documentclass{article}
\usepackage[paper=A4,pagesize,DIV=15]{typearea}
\usepackage[nostamp]{draftwatermark}
\SetWatermarkScale{6}
\SetWatermarkColor[gray]{1}
\usepackage[OT1, T1]{fontenc}
\DeclareTextSymbolDefault{\DH}{T1}
\usepackage[utf8]{inputenc}
\usepackage[english]{babel}
\usepackage[backend=biber,bibencoding=utf-8]{biblatex}
\usepackage{xargs}
\usepackage{xifthen}
\usepackage[inline]{enumitem}
\setlist{nolistsep, noitemsep}
\usepackage{amsthm}
\usepackage{amsmath}
\usepackage{amssymb}
\usepackage{nccmath}
\usepackage{calc}
\usepackage{array}
\usepackage{afterpage}
\usepackage[ruled,noresetcount]{algorithm2e}
\usepackage{hyperref}
\usepackage[capitalize]{cleveref}
\usepackage[dvipsnames]{xcolor}
\definecolor{GrayTableHighlight}{gray}{0.9}
\crefname{line}{line}{lines}
\usepackage{lipsum}
\newcommand{\colourfakemessage}[1]{\textcolor{BrickRed}{#1}}

\newcommand{\mmessage}[1]{{\color{black}\normalfont\ensuremath{\langle #1 \rangle}}}
\newcommand{\smessage}[2]{{\color{black}\normalfont\ensuremath{\langle #1 \rangle_{\sigma_{#2}}}}}
\newcommand{\tsmessage}[3]{\textnormal{\color{black}\ensuremath{\langle \textsf{#1}, #2 \rangle_{\sigma_{#3}}}}}
\newcommand{\tsmessageinnsmessage}[4]{\textnormal{\color{black}\ensuremath{\langle\langle \textsf{#1}, #2 \rangle_{\sigma_{#3}}, #4\rangle}}}
\newcommand{\tsmessageinnnsmessage}[3]{\color{black}\textnormal{\ensuremath{\langle\langle \textsf{#1}, #2 \rangle, #3\rangle}}}
\newcommand{\tnsmessage}[2]{\textnormal{\color{black}\ensuremath{\langle \textsf{#1}, #2 \rangle}}}
\newcommand{\tfakensmessage}[2]{\textnormal{\colourfakemessage{\ensuremath{\langle \textsf{#1}, #2 \rangle}}}}
\usepackage[usestackEOL]{stackengine}
\usepackage{scalerel}
\def\myoverline#1{\ThisStyle{%
		\setbox0=\hbox{$\SavedStyle#1$}%
		\stackengine{1.4\LMpt}{$\SavedStyle#1$}{\rule{\wd0}{.8\LMpt}}{O}{c}{F}{F}{S}%
	}}
\newcommand{\fv}[1]{{\myoverline{#1}}}

\newcommand{\kec}[1]{\textnormal{\ensuremath{\mathnormal{\texttt{KEC}(#1)}}}}
\newcommand{\cs}[1]{\ensuremath{\textit{CS}(#1)}}

\newcommand{\quorum}[1]{\ensuremath{\textit{Quorum}(#1)}}

\newcommand{\fb}{\ensuremath{\mathit{FB}}}
\newcommand{\pb}[1]{\ensuremath{\mathit{PB}}}
\newcommand{\eb}{\ensuremath{\mathit{EB}}}

\newcommand{\nprepared}{\emph{prepared}}
\newcommand{\nprepare}[1]{\emph{prepare#1}}
\newcommand{\ibfp}[1]{IBFT-block-finalisation-protocol}
\newcommand{\ibftmone}[1]{IBFT-M1}
\newcommand{\ibfpmone}[1]{\ibfp{}-M1}
\newcommand{\ibftmonepmone}[1]{IBFT-M1-\ibfpmone{}}
\newcommand{\ibfpmtwo}[1]{\ibfp{}-M2}
\newcommand{\ibftmonepmtwo}[1]{IBFT-M1-\ibfpmtwo{}}
\newcommand{\ibftmonepmthree}[1]{IBFT-M1-\ibfpmthree{}}
\newcommand{\ibfpmthree}[1]{IBFT-2.0-block-finalisation-protocol}
\newcommand{\citewithauthor}[1]{\citeauthor*{#1}~\cite{#1}}
\newcommand{\devptwop}[1]{\textsf{\DH}\ensuremath{\mathsf{ \Xi Vp2p}}}
\newcommand{\bconcat}{\mathop{\kern 0.05em||}}

\newtheorem{theorem}{Theorem}
\newtheorem{lemma}{Lemma}
\newtheorem{corollary}{Corollary}
\numberwithin{corollary}{lemma}

\theoremstyle{definition}
\newtheorem{definition}{Definition}

\newtheorem{assumption}{Assumption}

\newcounter{improvement}
\DeclareRobustCommand{\improvement}{\refstepcounter{improvement}Improvement~\theimprovement}
\crefname{improvement}{Improvement}{Improvements}

\newcommand{\iref}[1]{\hyperref[#1]{\ref*{#1}}}

\title{IBFT 2.0: A Safe and Live Variation of the IBFT Blockchain Consensus Protocol for Eventually Synchronous Networks}

\author{Roberto Saltini
		\and 
		David Hyland-Wood
	}

\date{\textit{PegaSys (ConsenSys)} \\
	\vspace{1em}
	\today}

\begin{document}

\SetKwBlock{Init}{Initialisation:}{}
\SetKwBlock{Functions}{Functions:}{}
\SetKwBlock{Expansions}{Macro Expansions:}{}
\SetKwBlock{Procedures}{Procedures:}{}
\SetKwBlock{UponRules}{Upon Blocks:}{}
\SetKwData{Reception}{reception of}
\SetKwData{HasReceived}{has received}
\SetKwFunction{Proposer}{proposer}
\newcommand{\fairproposer}[1]{\ensuremath{FairProposer(#1)}}
\SetKwFunction{Sign}{sign}
\SetKwFunction{EcRecover}{EthAddressRecover}
\SetKwFunction{VerifySignature}{verifySignature}
\SetKw{Let}{let}
\SetKwFunction{AV}{validators}
\SetKwFunction{Valid}{isValidBlock}
\SetKwFunction{ValidNonFinalised}{isValidNonFinalisedBlock}
\SetKwFunction{N}{n}
\SetKwFunction{NumOf}{numOfReceived}
\SetKwFunction{Received}{received}
\SetKwFunction{TimeoutForRoundZero}{timeoutForRoundZero}
\SetKwFunction{ReceivedMessages}{messagesReceivedLike}
\SetKwFunction{SetOfReceivedMessages}{setsOfMessagesReceivedLike}
\SetKw{OfSize}{of size}
\SetKw{Broadcast}{broadcast}
\SetKw{NotIncluding}{excluding}
\SetKw{Multicast}{multicast}
\SetKw{Send}{send}
\SetKw{True}{true}
\SetKw{False}{false}
\SetKwFunction{WellFormedToAddFinalisationProof}{wellFormedToAddFinalisationProof}
\SetKwFunction{RoundTimerTimeout}{roundTimerTimeout}
\SetKwFunction{CreateNewProposedBlock}{createNewProposedBlock}
\SetKwFunction{BlockHeight}{blockHeight}
\SetKwFunction{ExtractBlock}{extractBlock}
\SetKwFunction{ExtractFinalisationProof}{extractFinalisationProof}
\SetKwFunction{ExtractRound}{extractRound}
\SetKw{From}{\:from\:}
\SetKw{With}{\:such that\:}
\SetKw{IsSuchThat}{\:is such that\:}
\SetKw{IsSuchThatc}{\:is such that:\:}
\SetKw{Withc}{\:such that:\:}
\SetKw{Stop}{stop}
\SetKwFunction{AllSubsetsOf}{allSubSetsOf}
\SetKwFunction{UnionOfAllSetsIn}{unionOfAllSetsIn}
\SetKw{Any}{any\:}
\SetKw{Start}{start}
\SetKw{Set}{set}
\SetKw{Expiry}{expiry of}
\SetKw{And}{\:and\:}
\SetKw{Or}{\:or\:}
\SetKw{In}{\:in\:}
\SetKw{Not}{\:not\:}
\SetKw{For}{for\:}
\SetKw{Ifk}{if\:}
\SetKw{Choose}{choose any in\:}
\SetKw{ChooseSet}{choose any set in\:}
\SetKw{Implies}{implies\:}
\SetKw{ThereNotExists}{there not exists\:}
\SetKw{ThereExists}{there exists\:}
\SetKw{Nil}{null}
\def\uponcolor {Blue}
\SetKwFor{Upon}{upon}{do}{}
\SetKwProg{Fn}{}{$\equiv$}{}
\SetKwIF{FnInline}{FnInline2}{FnInline3}{}{$\equiv$}{}{}{}
\SetKwProg{FnPhy}{def}{:}{}
\SetKwProg{Proc}{def}{:}{}
\SetKwFunction{SizeOf}{sizeOf}
\SetKwIF{Expand}{Expand2}{Expand3}{}{expands to:}{}{}{}
\newcommand{\avibfp}[1][]{\AV{\ensuremath{{chain}_v[0:h_{#1}-1]}}}
\newcommand{\nmacro}[1][]{\N{\ensuremath{{chain}_v[0:h_{#1}-1]}}}
\newcommand{\nhmacro}[1][]{\ensuremath{n_{h,v}}}
\newcommand{\avhmacro}[1][]{\ensuremath{{\AV}_{h,v}}}
\newcommand{\fairproposermacro}[1][]{\fairproposer{\ensuremath{{chain}_v[0:h-1],  r_{h,v}}}}
\newcommand{\proposermacro}[1][]{\Proposer{\ensuremath{{chain}_v[0:h-1],  r_{h,v}}}}
\newcommand{\proposermacrop}[1]{\Proposer{\ensuremath{{chain}_v[0:h-1],  #1}}}
\newcommand{\hproposer}[1][]{\ensuremath{{\Proposer}_{h,v}(r_{h,v})}}
\newcommand{\hproposerp}[1]{\ensuremath{{\Proposer}_{h,v}(#1)}}
\LinesNumbered

\SetAlgoNoLine
\SetAlgoHangIndent{2em}
\SetNlSty{textbf}{}{:}
\newcommand\mycommfont[1]{\ttfamily\textcolor{OliveGreen}{#1}}
\SetCommentSty{mycommfont}
\DontPrintSemicolon

\maketitle

\begin{abstract}
In this work, we present IBFT 2.0 (Istanbul BFT 2.0), which is a Proof-of-Authority (PoA) Byzantine-fault-tolerant (BFT) blockchain consensus protocols that (i) ensures immediate finality, (ii) is robust in an eventually synchronous network model and (iii) features a dynamic validator set.
IBFT 2.0, as the name suggests, builds upon the IBFT blockchain consensus protocol  retaining all of the original features while addressing the safety and liveness limitations described in one of our previous works. 
In this paper, we present a high-level description of the IBFT 2.0 protocol and related robustness proof.
Formal specification of the protocol and  related formal proofs will be subject of a separate body of work.
We also envision a separate work that will provide detailed implementation specifications for IBFT 2.0.
\end{abstract}

\section{Introduction}

%
%

\vspace{1em}
\vspace{1em}
IBFT 2.0 is Proof-of-Authority (PoA) Byzantine-Fault-Tolerant (BFT) blockchain consensus protocol that enables consortium network to leverage on the capabilities of Ethereum smart contracts, ensures immediate finality, is robust in an eventually synchronous network and features a dynamic validator set.
As the name suggests, the IBFT 2.0 protocol is based on the IBFT protocol that was was developed around early 2017 by AMIS Technologies \cite{eip650} and was fully implemented in Quorum \cite{jpmorgangithub} by around November 2017.
The IBFT protocol features all of the properties of the IBFT 2.0 protocol mentioned above except for robustness in eventually synchronous networks as identified by \citewithauthor{IBFT1-Analysis}.
IBFT 2.0 addresses the robustness issues of the IBFT protocol while maintaining all of its original properties.
In the following section we describe each of these properties in more detail.

\subsection{Properties of the IBFT 2.0 protocol}

\paragraph{Blockchain consensus protocol.}
Blockchains are the most widely adopted implementations of \emph{distributed ledgers} which are append-only databases of transactions that are replicated across multiple participants, hereafter called \emph{nodes}.
The trust and responsibility for maintaining the database is spread across all of the nodes or a subset of them. 
This is in contrast to a traditional centralised system where full trust is given to a central authority responsible for maintaining the database.
One of the issues with this traditional approach is that the central authority has the power to alter the database unilaterally.
The decentralisation aspect of distributed ledger technology makes it well suited to any use case where the need for a central authority is either adding costs to the system or undermining the trust in the system itself.
Blockchains implements distributed ledgers by batching transactions into blocks and cryptographically linking each block to the previous one forming a chain of blocks, which is where the technology takes its name from.
Consensus protocols play a fundamental role in the blockchain technology as they have the responsibility to ensure that the chain of blocks replicated amongst the nodes is consistent.
The type of network and environment assumptions made when designing a consensus protocol influence how the blockchain performs once deployed in a real environment and network.
Some of the key performance metrics that are heavily influenced by consensus protocols are: (i) throughput or number of transactions per second, (ii) latency or time taken from when a transaction is submitted to the system to when the transaction is included in a block and (iii) robustness or what type of attacks the protocol can withstand.

\paragraph{Ethereum smart contracts.}
Compared to Bitcoin, which was the first widely adopted blockchain and mainly allows transferring values between nodes, the Ethereum blockchain specifies a Turing-complete language that can be used to  build small distributed programs, called \emph{smart contracts}, that are executed in a sandboxed runtime by each node.
The runtime, called the Ethereum Virtual Machine (EVM) modifies the Ethereum global state maintained by each node. 
This means that any user of the Ethereum blockchain has the capability to create \emph{smart contracts} that can govern the interaction between the different users of the system and automate value transfer in a decentralised fashion.
One of the first use cases for Ethereum was the creation of escrow smart contracts eliminating the need for a trusted 3rd party.

\paragraph{Byzantine-fault-tolerant (BFT).}
Byzantine-fault-tolerant, or BFT, specifies the type of node fault mode that the consensus protocol can cope with.
Specifically, BFT identifies a class of blockchain consensus protocols that ensure blockchain consistency despite some of the nodes, referred to as Byzantine, being malicious and acting arbitrarily.
The usage of the word Byzantine to identify malicious nodes dates back to the paper \citetitle{Lamport:1982:BGP:357172.357176} by \citewithauthor{Lamport:1982:BGP:357172.357176}.
The Byzantine failure mode is the strongest failure mode considered in the consensus protocol literature.
Another common but weaker failure mode is \emph{fail-stop} failure mode which only considers nodes stopping communicating but never acting maliciously.

\paragraph{Proof-of-Authority (PoA).}
Another way to classify consensus protocols is by the technique used to prevent an attacker from conducting a \emph{Sybil attack} which consists in one node being able to gain power in the system by creating multiples pseudonymous identities.
Typically, creating a new digital identity that can be used to interact with a blockchain is quite cheap as it just requires generating a random private key and the related public key which can be done in a matter of few seconds on any modern personal computer.
One the most widely used and famous techniques for preventing Sybil attacks is \emph{proof of work} which was originally pioneered by \citewithauthor{Dwork:1992:PVP:646757.705669} and gained subsequent publicity by its employment in Bitcoin. 
Proof-of-work (PoW) requires node to spend compute effort in solving a hard cryptographic puzzle before being able to propose a block to be added to the blockchain.
Proof-of-Stake (PoS) is another quite well known technique where the right to propose new blocks is given according to the amount of stake owned. 
In contrast, in Proof-of-Authority, or PoA, Sybil attacks are prevented by conferring the right to create new blocks only  to a defined set of nodes.
Within the IBFT 2.0 protocol, the nodes with the right to create new blocks and ensuring blockchain consistency are called \emph{validators}.

\paragraph{Consortium Blockchain.}
Compared to permissionless, or public, blockchains, like Bitcoin or Ethereum, where anybody can join the network and participate in the protocol, in consortium networks there exists some level of permissioning which enables only a set of nodes in proposing new blocks and participating in the consensus protocol.
It should be quite evident why PoA type consensus protocols, like the IBFT protocol family, are well suited to this type of permissioning.
It should be noted that while not every node can propose new blocks, some  consortium blockchains allow any node to read data from the blockchain.
The IBFT protocol family also affords for this type of configuration.

\paragraph{Immediate Finality.}
A transaction is defined as final once it is included in the blockchain and it can not be removed from it or changed of position except if the environment is compromised, which, for example, can occur if the number of Byzantine nodes is higher than the maximum number that the protocol can withstand.
Immediate finality means that as soon as a transaction is included in a block, the protocol guarantees that it will not be removed or changed of position.
As comparison, the PoW consensus protocol in Bitcoin and Ethereum only guarantees eventual probabilistic finality where the deeper a transaction is in the blockchain, the less probable is that the transaction can be removed or changed of position.
As further comparison, Casper FFG \cite{CasperFFG}, the Ethereum 2.0 PoS consensus protocol, provides eventual ``non-probabilistic'' finality, where transaction will eventually reach a state where they cannot be removed or moved of position but this does not necessarily happen at every block.
In IBFT 2.0 finality is immediate.


\paragraph{Robustness.}
Our definition of robustness for the IBFT protocol family is based on the definition of robustness for public transaction ledgers provided in \citewithauthor{GKL:2018}.\\
For the purpose of this definition, the position of a transaction within the transaction ledger implemented by the IBFT 2.0 protocol is defined as a pair with the first component corresponding to the height of the block including the transaction and the second component corresponding to the position of the transaction within the block.
\begin{definition}\label{def:robustnes}
	A blockchain consensus protocol implements a robust distributed permissioned transaction ledger with immediate finality and $t$-Byzantine-fault-tolerance  if, provided that no more than $t$ validators are Byzantine, it guarantees the following two properties:
	\begin{itemize}
		\item \textbf{Persistence.} If an honest node adds transaction $ T $ in position $ i $ of its local transaction ledger, then (i) $ T $ is the only transaction that can ever be added in position $i$ by any other honest node, (ii) $T$ will eventually be added to the local transaction ledger of any other honest node.
		
		\item \textbf{Liveness.} Provided that a transaction is submitted to all honest validators, then the transaction will eventually be included in the local transaction ledger of at least one honest node.
	\end{itemize}
\end{definition}

\paragraph{Eventually synchronous network.}
In the consensus protocol literature there are three main network models that have been considered which differ on the assumption made regarding transmission latency:
\begin{itemize}
	\item Synchronous network: the maximum latency (time required for a message to reach the recipient) is bounded and known;
	\item Asynchronous network: the maximum latency is unknown and messages may never be delivered;
	\item Partially synchronous network: this model, which was first introduced by \citewithauthor{Dwork:1988:CPP:42282.42283}, lies in between the other two.
	Specifically, there are two possible definitions of partial synchrony:
	\begin{itemize}		
		\item Messages are guaranteed to be delivered but the maximum latency, while finite, is unknown. To the best of our knowledge, no specific name has been defined for this model;		
		\item Eventually synchronous network: there exists a point in time, called global stabilisation time, or GST, after which the message delay is bounded by a finite and constant value;
	\end{itemize}
\end{itemize}
The model with the weakest assumptions is the asynchronous network model, followed by the partially synchronous network model and the synchronous network model in this order.
Between the two definitions of partial synchrony, eventual synchrony is the one with the weakest assumptions.
As proved by \citewithauthor{FLP} in \citeyear{FLP}, no consensus protocol that aims to tolerate at least one fail-stop node is guaranteed to terminate in the asynchronous model (the one with the weakest assumption).
There exist solutions that operate in the asynchronous network model, but the termination is only guaranteed probabilistically \cite{Bracha:1985:Asynchronous-Consensus-and-Broadcast-Protocols, honeybadger}.
The IBFT 2.0 protocol guarantees deterministic termination of the sub-protocol that has the responsibility to decide on the blocks to be added to the blockchain which means that IBFT 2.0 guarantees that the blockchain does not stop growing.
However, as discussed in the robustness proof section, a probabilistic assumption is required to show that for any block there is eventually a block created successively to it which is not empty. 
This second condition is important for the Liveness property stated above.
In \cref{sec:imp:remove-prob-assumption}, we discuss how the IBFT 2.0 protocol can be modified to drop any probabilistic assumption and achieve deterministic robustness under the eventually synchronous network model which is the weakest assumption where termination can be deterministically guaranteed.

\paragraph{Dynamic validator set.}
Compared to classic (non-blockchain) consensus protocols like PBFT \cite{Castro:1999:PBF:296806.296824}, where the set of protocol nodes is known in advance and never changes, IBFT 2.0, like Clique, allows the nodes to add and remove validators by a voting mechanism.

\subsection{Our Contribution}
As mentioned above, IBFT 2.0 builds upon the IBFT consensus protocol \cite{eip650,jpmorgangithub}  addressing the following limitations of IBFT as described by \citewithauthor{IBFT1-Analysis}.
\begin{itemize}
	\item Persistence is not guaranteed.
	\begin{itemize}
		\item One Byzantine validator can potentially be able to remove or change the position of a transaction that has already been finalised.
		\item In a network of six validators, even if all validators are honest, a network partitioning can cause the blockchains maintained by two sets of three validators each to diverge. Once the partitioning is resolved, validators have to choose one chain which means removing or reordering the transactions of the other chain. 
		\item IBFT does not guarantee that a transaction added to the local blockchain of one validator is eventually added to the local blockchain of all other nodes.
	\end{itemize}
	\item Liveness is not guaranteed. 
	Specifically, the IBFT protocol may reach a state where no new blocks can be added to any local blockchain which means that no new transactions can be added to the distributed ledger.
\end{itemize}

While \citetitle{IBFT1-Analysis}~\cite{IBFT1-Analysis} proposes a clear solution for addressing the issues with the Persistence property of the IBFT protocol, it only sketches two possible solutions to the Liveness issue without providing a full description of the resulting protocol.
As a contribution, this work defines a complete solution to the Persistence and Liveness issues identified in \cite{IBFT1-Analysis} and provides related robustness proof.

In this work we present a protocol-level model of the IBFT 2.0 consensus protocol well suited for reasoning about the robustness of the protocol by abstracting the details of the actual implementation \cite{pantheongithub}.
We envision releasing a future body of work describing the details of the implementation (e.g. precise encoding of the messages, precise definition of the header structure) which can be used to create interoperable implementation of the protocol in any Ethereum client.
We also envision to produce a more formal definition of the protocol using formal languages (e.g. TLA+, Verdi, Alloy) and use formal proof systems to prove the properties of the protocol.


This paper is organised as follows. 
In Section \ref{sec:system-model} we present our analysis model. 
In \cref{sec:protocol-description} we  describe the IBFT 2.0 protocol as implemented in Pantheon \cite{pantheongithub}.
In \cref{sec:robustness-analysis} we present the robustness analysis of the IBFT 2.0 protocol and in \cref{sec:improvements} we present a series of improvements that can applied to the IBFT 2.0 protocol.

\newcommand{\bytes}[1]{\ensuremath{\mathbb{B}_{#1}}}
\newcommand{\nil}[1]{\ensuremath\bot}

\section{System Model}\label{sec:system-model}
The system model considered for IBFT 2.0 is the same system model considered in the analysis of IBFT \cite{IBFT1-Analysis}.
We re-state the properties of the model here for convenience.

\paragraph{Asynchronous nodes.}
We consider a system composed of an unbounded number of asynchronous nodes, each of them maintaining a local copy of the blockchain obtained by adding blocks to it as specified by the IBFT 2.0 protocol.
We assume that all nodes have the same genesis block.

\paragraph{Network Model.}
The IBFT 2.0 protocol relies on the Ethereum \devptwop{} protocol for the delivering of all protocol messages.
We model the gossip network as an eventually synchronous network, as defined in \citewithauthor{Dwork:1988:CPP:42282.42283}, where there exists a point in time called global stabilisation time (GST), after which the message delay is bounded by a constant, $\Delta$.
Before GST there is not bound on the message delay and we admit messages being lost.

\paragraph{Failure Model.}
We consider a Byzantine failure mode system, where Byzantine nodes can behave arbitrarily. In contrast, honest nodes never diverge from the protocol definition.
We denote the maximum number of Byzantine nodes that an eventually synchronous network of $n$ nodes participating in the consensus protocol can be tolerant to with $f(n)$.
As proved in \citewithauthor{Dwork:1988:CPP:42282.42283}, the relationship between the total number of nodes, $n$, and the maximum number of Byzantine nodes can be expressed as follows:
\begin{equation}
f(n) \equiv \left \lfloor \frac{n-1}{3} \right \rfloor
\end{equation}

\paragraph{Cryptographic Primitives.}
The IBFT 2.0 protocol uses the Keccak hash function variant as per Ethereum Yellow Paper \cite{yellowpaper} to produce digests of blocks.
We assume that the Keccak hash function is collision-resistant.\\
The IBFT 2.0 protocol relies on the Elliptic Curve Digital Signature scheme already used in the Ethereum protocol to sign transactions.
We assume that this signature scheme ensures uniqueness and unforgeability.
Uniqueness means that the signatures generated for two distinct messages are different with high probability.
The unforgeability property ensures that Byzantine nodes, even if they collude, cannot forge digital signatures produced by honest nodes.\\
We use \smessage{m}{v} to denote a message $m$ signed by validator $v$.

\section{Protocol Description}\label{sec:protocol-description}

In this section we provide a description of the IBFT 2.0 protocol.
We write ``IBFT 2.0 (IBFT)'' to indicate sections of the descriptions that apply to both IBFT and IBFT 2.0. 

As common to any blockchain implementation, each IBFT 2.0 (IBFT) node maintains a local copy of the blockchain where the first block, called \emph{genesis block}, is the same for all nodes.
Each block $B$ added to the blockchain must be cryptographically linked to another block in the blockchain, $B_p$ which is commonly defined as the \emph{parent} of block $B$, and, conversely, $B$ is defined as the child of $B_p$.
In IBFT 2.0 (IBFT), starting from the genesis block, the next block to be added to the local blockchain maintained by a node is the child of the latest block that was added to the blockchain.
As it may be evident by now, the IBFT 2.0 (IBFT) blockchain can be modelled as a linked list of blocks, rather than  a tree like the public Ethereum blockchain.
In alignment with the terminology used in literature, the height of a block is defined as the number of parent links separating the block from the genesis block which has height 0.

The IBFT 2.0 (IBFT) protocol can be modelled as running sequential instances of what we call the \emph{\ibfpmthree{}},
where the objective of the $h$-th instance of the \ibfpmthree{} is to decide which Ethereum block, and consequently which set of transactions, are to be added at height $h$ of the blockchain maintained by any IBFT 2.0 (IBFT) node.
Only a subset of the entire set of IBFT 2.0 (IBFT) nodes can participate in the $h$-th instance of the block finalisation protocol. 
We call this set of nodes the \emph{validators for height/instance $h$} and refer to each member of this set as a \emph{validator for height/instance $h$}.
We also refer to all of the nodes not included in the validator set for height/instance $h$ as \emph{standard nodes}.
We often omit \emph{for height/instance $h$} when this is clear from the context.
The set of validators for each instance $h$ of the \ibfpmthree{} is deterministically computed as function of the chain of blocks from the genesis block until the block with height $h-1$.\\
As explained in more detail in the following section, each instance of the \ibfpmthree{} is organised in rounds and in each round one of the validators is given the responsibility to propose an Ethereum block for the height associated with the specific instance of the \ibfpmthree{} that the validator is running.
Once agreement is reached, the \ibfpmthree{} creates a \emph{finalised block} which includes the Ethereum block and additional information that allows any node, even nodes that did not participate in the \ibfpmthree{}, to verify that agreement on the Ethereum block included in the finalised block was correctly reached.

In practice, each IBFT 2.0 (IBFT) node adds finalised blocks to its local blockchain, not only the Ethereum blocks included in them.
In this way, any node joining the network at any point in time, when synching its local blockchain with its peers, receives all the information required to verify that agreement was indeed reached correctly on each block that it receives, even on those created before it joined the network.
Each IBFT 2.0 finalised block \fb{} can be modelled by the tuple $(\fb_{EB},\fb_{FP})$ where $\fb_{EB}$ is the Ethereum block to be added to the blockchain, $\fb_{FP}$ is the proof that  agreement was correctly reached on the position in the chain of the block $\fb_{EB}$.
\\
Each finalisation proof $FP$ can be in turn modelled by the tuple $(FB_{r}, FB_{CS})$ where $FB_{r}$ is the round number of the \ibfpmthree{} during the execution of which agreement on the block inclusion in the blockchain was reached and $\fb_{CS}$ is a list of signatures on both the Ethereum block and the round proving that agreement was indeed reached by a correct execution of the \ibfpmthree{}. 
More detail on how this list of signatures, called \emph{commit seals}, are computed is presented in \cref{sec:description-of-the-ibfpmthree-protocol}.

Each Ethereum block can carry a vote, cast by the proposer of that block, to add a validator to or remove a validator from the validator set.
Once more than half of the validators cast a consistent vote to add or remove a validator to/from the validator set, the validator is added or removed from the validator set starting from the next block and all of the votes targeting this validator are discarded.
In this paper we do not provide a pseudocode description of this algorithm, but we may add it to a future revision of this work.

The IBFT 2.0 consensus protocol is described by the pseudocode in \Cref{algo:ibft-protocol,algo:ibpfpmone} where:
\begin{itemize}
	\item Statements are expressed in a mathematical form but with standard mathematical symbols replaced by their equivalent English version, e.g we write \In rather than $\in$, $\And$ rather than $\land$, \ThereExists rather than $\exists$ and so on.
	Our intent is to provide an unambiguous definition of the protocol which can be understood by people that are not familiar with standard mathematical notation.
	Also, comments identified by \CommentSty{text in typewriter font and green colour} are used to provide natural language description of pseudocode statements which may not be immediately obvious;  
	\item For brevity of notation, we avoid using individual existential quantifiers (i.e $\exists$ in mathematical notation and \ThereExists in ``English notation'') for message fields but rather express existential quantifier on the entire message.
	We use an overhead line (e.g $\fv{var}$), to indicate message fields that, if the extensive notation was used, then they should be expressed via an existential quantifier.
	For example $\ThereExists \mmessage{f_1,\fv{f_2}} \In receivedMessages_v $ stands for $\ThereExists f_2 \With \ThereExists \mmessage{f_1,f_2} \in receivedMessages_v $;	
	\item Each of the \textbf{upon} blocks in the pseudocode is assumed to be executed atomically when the condition specified after the \textbf{upon} keyword is satisfied;
	\item All functions in \FuncSty{typewriter font} are defined in the remainder of this section, whereas all functions in \textit{italic font} are defined in the pseudocode;
	\item $ receivedMessages_v $ corresponds to the set of all messages received by node $v$;
	\item $ \FuncSty{peers}_v  $ corresponds to the set of \devptwop{} Gossip protocol peers of $v$;
	\item $ \{ m \in V \With P(m) \} $ corresponds to the set of all the elements of $V$ for which predicate $P$ is true;
	\item $ \{ F(m) \With m \in V \And P(m) \} $ corresponds to the set obtained by applying the function $ F $ to all the elements of $ V $ for which predicate $P$ is true;
	\item $\AllSubsetsOf{M}$ corresponds to the set of all of the subsets of $M$ which is normally called the power set of $M$.
	For example, $\AllSubsetsOf{\ensuremath{\{m_1,m_2,m_3\}}}$ corresponds to the set $\{ \{\}, \allowbreak \{m_1\},\allowbreak \{m_2\}, \allowbreak\{m_3\}, \allowbreak\{m_1,m_2\}, \allowbreak\{m_1,m_3\}, \allowbreak\{m_2,m_3\}, \allowbreak \{m_1,m_2,m_3\}  \}$;	
	\item The symbol $*$ denotes any value;
	\item \colourfakemessage{Dark red} colour denotes messages used only for modelling purposes and that do not have an immediate one-to-one relationship with the messages of the ETH sub-protocol;
	\item Black colour when applied to messages denotes IBFT 2.0 specific messages not included in the current ETH sub-protocol;
	\item We use the notation $T: (t_1, \ldots, t_n)$ to indicate a tuple $(t_1, \ldots, t_n)$ that we successively refer to as $T$; 
	\item $\pi_m(T)$ corresponds to the $m$-th element of the tuple T where the first element has index 1.
	For example, $\pi_2((t_1, t_2, t_3))$ corresponds to $t_2$.
	\item $\BlockHeight{\eb}$ is defined as the height of the Ethereum block $\eb$;
	\item For clarity of notation, we use \SizeOf{$ M $} to indicate the size of the set $M$, i.e $\SizeOf{\ensuremath{M}} \equiv \|M\|$;
	\item \EcRecover{\ensuremath{H,signature}} corresponds to the Ethereum address whose signature of the hash $H$ corresponds to $signature$;
	\item Each validator $v$ stores its local blockchain in ${chain}_v$;
	\item ${chain}_v[n]$ corresponds to the finalised block with height $n$, while ${chain}_v[n:m]$ corresponds to a sub-chain including all of the finalised blocks from height $n$ to height $m$ included;
	\item ${chain_{EB}}_v[n]$ corresponds to the Ethereum block included in the finalised block with height $n$, while ${chain_{EB}}_v[n:m]$ corresponds to a sub-chain including all of the Ethereum blocks included in the finalised blocks from height $n$ to height $m$ included;
	\item The blockchain height is defined as the height of the last finalised block added to the blockchain;
	\item  \AV{${chain}_v[0:h-1]$} represents the set of authorised validators for instance $h$ of the \ibfpmthree{};
	\item  $\N{\ensuremath{{chain}_v[0:h-1]}}$ represents the number of validators for instance $h$ of the \ibfpmthree{}, i.e. $\N{\ensuremath{{chain}_v[0:h-1]}} \equiv \SizeOf{\AV{\ensuremath{{chain}_v[0:h-1]}}} $.
	\item The function $\Valid{\ensuremath{\eb,\eb_{p}}}$ is defined to be true if and only if block $\eb$ is a valid Ethereum block with parent $\eb_{p}$. 
	For the purpose of this work, we consider that \linebreak $\Valid{\ensuremath{\eb,\eb_{p}}}$ only verifies the following fields of the standard Ethereum header: parentHash, stateRoot, transactionsRoot, receiptsRoot, logsBloom, number, gasLimit, gasUsed.
	These fields are verified as specified in \cite{yellowpaper}.
	The IBFT 2.0 (IBFT) protocol implementation actually verifies also the other fields but in a different way than specified in \cite{yellowpaper}.
	We do not describe how these fields are verified as this is out of the scope of this work and does not affect our results.
	These details will be discussed in a future document describing the implementation details of the IBFT 2.0 protocol.
\end{itemize}

\begin{algorithm}[b!]
	\setcounter{AlgoLine}{0}
	\Functions{ 
		\lFnInline{\quorum{n}}{
			\begin{minipage}[h]{10cm}
				\begin{flalign*}
				\left \lceil \frac{2n}{3} \right \rceil&&
				\end{flalign*}
			\end{minipage}
		}
		\BlankLine
		\Fn{${isValidFinalisedBlock}(\fb{}:(\fb_{EB},\fb_{FP}:(\fb_{{FP}_{r}},\fb_{{FP}_{cs}})),v)$\label{ln:isValidFinalisedBlock}}{
			\mbox{\hspace{1.9em}\quad}$\SizeOf{\ensuremath{{FP}_{cs}}} \geq \quorum{\nhmacro[]} \And \newline
			\mbox{\quad}\Any cs \in {FP}_{cs} \Withc \newline
			\mbox{\quad\quad}\EcRecover{\ensuremath{\kec{(\fb_{EB},\fb_{{FP}_{r}})},cs}}\In \avibfp[v] \newline$
		}		
	}
	\BlankLine
	\Init{
		${chain}_v[0] \gets \textit{genesis block}$\;
		$h_v \gets 1 $\;
		\lIf{$v \In \avibfp[v] $}{\Start the $h_v$-th instance of the \ibfpmthree{}}
	}
	\BlankLine
	\UponRules{	 
		\Upon{\color{\uponcolor}\label{ln:finalise-block-start}
			$ \tfakensmessage{FINALISED-BLOCK}{\fv{\fb{}}:(\fv{\fb_{EB}},\fv{\fb_{FP}})} \In receivedMessages_v $ } {
			\If{$\BlockHeight{\ensuremath{\fv{\fb_{EB}}}} = h_v $}{
				\If{${isValidFinalisedBlock}(\fv{\fb{}}),v)$ } {
					${chain}_v[h_v] \gets \fv{\fb{}}$\;
					\lIf{$v \In \avibfp[v] $}{\Stop the $h_v$-th instance of the \ibfpmthree{}}
					$h_v \gets h_v + 1$\;
					\lIf{$v \In \avibfp[v] $}{\Start the $h_v$-th instance of the \ibfpmthree{}}
				}
			}
		}
		\Upon{\color{\uponcolor}\label{ln:request-blocks-when-message-with-higher-height-is-received}
			\begin{tabular}[t]{@{}l@{\quad}l}
				$ ( $\\
				&\tsmessageinnsmessage{PROPOSE}{\fv{h_m}, *, *}{\fv{sender}}{*,*} $\In receivedMessages_v $ \Or \\
				&\tsmessage{PREPARE}{\fv{h_m}, *, *}{\fv{sender}} $\In receivedMessages_v $ \Or \\
				&\tsmessage{COMMIT}{\fv{h_m}, *, *}{\fv{sender}} $\In receivedMessages_v $ \Or \vspace{0.3em}\\
				&\tsmessageinnsmessage{ROUND-CHANGE}{\fv{h_m}, *, *}{\fv{sender}}{*} $\In receivedMessages_v $\\
				\multicolumn{2}{@{}l}{$ ) \And $}\\
				\multicolumn{2}{@{}l}{$\fv{h_m} > h_{v} \And$}\\
				\multicolumn{2}{@{}l}{$ \fv{sender} \in \FuncSty{peers}_v \And$}\\
				\multicolumn{2}{@{}l}{$ \FuncSty{expectedHeight}_v[\fv{sender}] < \fv{h_m }$}\\
			\end{tabular}\\
		} {
		$ \FuncSty{expectedHeight}_v[\fv{sender}] \gets \fv{h_m }$\;
		\Send \tfakensmessage{GET-BLOCKS}{h_v,\fv{h_m}}\;
	}		
}
\label{ln:finalise-block-end}	
\caption{IBFT 2.0 protocol  for IBFT node $v$}
\label{algo:ibft-protocol}
\end{algorithm}
			
As described by \cref{algo:ibft-protocol}, the different instances of the \ibfpmthree{} are started and stopped in the same way that instances are started and stopped in the IBFT protocol.
Finalised blocks are delivered to nodes using the standard ETH \devptwop{} sub-protocol.
In the pseudocode we abstract the actual ETH sub-protocol and model the reception of a finalised block \fb{} with the reception of a \tfakensmessage{FINALISED-BLOCK}{\fb} message.
No such message really exists in the implementation of the IBFT 2.0 protocol \cite{pantheongithub}.
As per standard ETH sub-protocol, a block can be received either via an unsolicited NewBlock message or via the pair of messages GetBlockHeaders/BlockHeaders followed by the pair of messages GetBlockBodies/BlockBodies \cite{ethsubprotocol}.

As described by the \textbf{upon} block at \cref{ln:finalise-block-start}, when a new finalised block is received by a node $v$, $v$ executes the following operations:
\begin{enumerate}
	\item verifies if the finalised block received is for the next expected chain height, i.e $ h_v $;
	\item if so, it verifies if the finalised block received is a valid finalised block.
	\item if both verifications pass, then:
	\begin{enumerate}[label*=\arabic*.]
		\item $v$ adds the finalised block to its local blockchain;
		\item if $v$ is a validator for the current instance of the \ibfpmthree{}, then $v$ stops that instance;
		\item $v$ advances the next expected block height, $h_v$, by 1;
		\item if $v$ is a validator for the \ibfpmthree{} instance for the new value of $h_v$, then $v$ starts that instance.
	\end{enumerate}
\end{enumerate}

As described by the function $isValidFinalisationBlock(\fb,v)$ at \cref{ln:isValidFinalisedBlock}, an IBFT 2.0 finalised block \fb{} is defined valid if and only if all of the following conditions are met:
\begin{itemize}
	\item it contains at least $\quorum{n}\equiv \left \lceil \frac{2n}{3} \right \rceil$ different commit seals, where $n$ is the number of validators for instance $h$ of the \ibfpmthree{}, i.e. $n \equiv \nmacro[]$;
	\item each commit seal corresponds to the signature of one of the validators over the Ethereum block and the round number included in the finalisation proof.
\end{itemize}

Compared to IBFT, IBFT 2.0 adds the \textbf{upon} block at \cref{ln:request-blocks-when-message-with-higher-height-is-received} to address the Persistence issue identified in Lemma 6 of \cite{IBFT1-Analysis}.
All of the IBFT 2.0 specific messages (i.e Proposal, Prepare, Commit and Round-Change) include the height of the \ibfpmthree{} instance that they relate to.
When a node $v$ receives one of these messages including a height $h_m$ with value $\geq$ than the next expected height $h_v$, if the sender of the message is one of the \devptwop{} peers of $v$, then $v$ starts asking the peer for finalised blocks with height between the $v$'s current height $h_v$ and the height $h_m$ included in the IBFT 2.0 message received.
We model this with the transmission of a \tfakensmessage{GET-BLOCKS}{h_v,h_m} message and model the peer's response to this request as a \tfakensmessage{FINALISE-BLOCK}{\fb} message.
As above, this is not an exact description of the real messages sent by the implementation.
It is a modelling of the \devptwop{} behaviour useful in this context for analysing the protocol.
$ \FuncSty{expectedHeight}_v[v']$ represents the blockchain height that node $v$ expects node $v'$ to have.
In our modelling of the protocol, we use $ \FuncSty{expectedHeight}_v[v']$ to express that Get-Blocks messages are sent only the first time that an IBFT 2.0 message with an height higher than $h_v$ is received.

\subsection{Description of the \ibfpmthree{}}\label{sec:description-of-the-ibfpmthree-protocol}
In this section we describe a generic $h$ instance of the \ibfpmthree{} for a validator $v$ as detailed in \cref{algo:ibpfpmone}.

While significant portions of the \ibfpmthree{} are similar to the \ibfp{} presented in Section 3 of \cite{IBFT1-Analysis}, for the sake of completeness, this section describes the full \ibfpmthree{}.
However, for those portions of the \ibfpmthree{} that are identical or very similar to the \ibfp{}, the description is taken almost verbatim from \cite{IBFT1-Analysis}.

As with the \ibfp{}, the \ibfpmthree{} is organised in rounds, starting from round 0, where validators progress to the next round once they suspect that in the current round they will not be able to decide on the block to be included at height $h$ of the blockchain.
Both in \Cref{algo:ibpfpmone} and here, the current round for the $h$-th instance of the \ibfpmthree{} for validator $v$ is denoted by $r_{h,v}$.

For each round, one of the validators is selected to play the role of block proposer. 
This selection is operated by the evaluation of \proposermacro[] where \Proposer{$ \cdot, \cdot $} is  a deterministic function of the chain of blocks from the genesis block until the block with height $h-1$ and the current round number.

The pseudocode at \cref{ln:avibfp,ln:hproposer,ln:nhmacro} introduces the following macros:
\begin{itemize}
	\item \nhmacro[]: number of validators for the $h$-th instance of the \ibfpmthree{} for validator $v$;
	\item \avhmacro[]: validators for the $h$-th instance of the \ibfpmthree{} for validator $v$;
	\item \hproposer[]: proposer for round $r_{h,v}$ of the $h$-th instance of the \ibfpmthree{} for validator $v$.
\end{itemize}
These macros are used both in the pseudocode and in this section to simplify the notation when describing the $h$-th instance of the \ibfpmthree{} for validator $v$. 
We use the term \emph{non-proposing validators for round $r$ and instance $h$} to indicate all of the validators for round $r$ and instance $h$ with the exclusion of the proposer for round $r$ and instance $h$.

For the purpose of this work, we do not define the proposer selection function, but we state that it ensures that all of the validators for the $h$-th instance of the \ibfpmthree{} are selected for any sequence of $\nhmacro[]$ consecutive rounds.
The IBFT 2.0 protocol retains the IBFT protocol capability to specify two alternative logics for selecting the proposer for round 0:
\begin{itemize}
	\item \textbf{Sticky Proposer.} The proposer for round 0 corresponds to the proposer of the block included in the previous finalised block;
	\item \hypertarget{def:round-roubin-selection-logic-definition}{\textbf{Round-Robin Proposer.}} The proposer for round 0 corresponds to the proposer of the proposer selection sequence that comes after the proposer of the block included in the previous finalised block.
\end{itemize}

Compared to IBFT, in IBFT 2.0 there is no block locking mechanism. 
The safety of the protocol is guaranteed by the round change protocol discussed below which is based on the view change protocol of PBFT~\cite{Castro:1999:PBF:296806.296824}.

As specified by the initialisation block (\cref{ln:initialisation}), if $v$ is the selected block proposer for the first round, i.e round $0$, then $v$ multicasts a Proposal message \tsmessageinnsmessage{PROPOSAL}{h, 0, \kec{\pb{}}}{v}{\pb{},\bot} to all validators (including itself) which comprises  the message \tsmessage{PROPOSAL}{h, 0, \kec{\pb{}}}{v} signed by $v$, the proposed block \pb{} and a Round-Change-Certificate which for round 0 is empty, i.e. $\bot$.
More detail on how the Round-Change-Certificate is assembled is provided further down in this section when discussing how validators move to a different round.
The proposed block \pb{} is modelled here as a tuple where the first element is a standard Ethereum block and the second element is the current round number at which the Ethereum block was created which at initialisation is 0.
\kec{\cdot} represents the Keccak hash function.
The pseudocode uses \CreateNewProposedBlock{h, v} to represent the creation of a new block with height $h$ by validator $v$. 
Honest validators employee a fair transaction selection algorithm to decide which transactions to include in the next block.
The definition of such algorithm is outside the scope of this work.

As specified by \crefrange{ln:upon-pre-prepare}{ln:upon-pre-prepare:end-accept}, a validator $v$ accepts a Proposal message \tsmessageinnnsmessage{PROPOSAL}{h_{pp}, r_{pp}, H}{\textit{PB}:(\eb,r_{\eb}),*} if and only if all of the following conditions are met:
\begin{itemize}
	\item $v$ is currently running the \ibfpmthree{} instance   $ h_{pp} $, i.e $h_{pp} = h$ ;
	\item $v$ is in round $0$, i.e. $r_{pp} = r_{h,v} = 0 $;
	\item the signed portion of the message, \tnsmessage{PROPOSAL}{h_{pp}, r_{pp}, H}, is signed by the selected proposer for round $r_{h,v} = 0$ and instance $h$ of the \ibfpmthree{};
	\item $v$ has not already accepted a Proposal message for round $ r_{h,v} = 0$ in the $h$-th instance of the \ibfp{}, i.e ${acceptedPB} = \bot$;
	\item the Ethereum block \eb{} included in the proposed block \pb{} is a valid block for height $h$;
	\item the round number included in the prepared block \pb{} matches the current round number, i.e. $r_{h,v} = r_{\eb} = 0$;
	\item $H$ corresponds to the Keccak hash of the proposed block \pb{}.
\end{itemize}
\begin{sloppypar}
	When a validator $v$ accepts a Proposal message:
\begin{itemize}
	\item  it marks the Proposal message as accepted by setting the state variable $ acceptedPB $ to the proposed block included in the Proposal message (\cref{ln:set-accepted-proposal});
	\item it multicasts a Prepare message \tsmessage{PREPARE}{h, r_{h,v}, H}{v} (see \cref{ln:broadcast-prepare}) to all validators (including itself).
\end{itemize}
\end{sloppypar}

\vspace{1em}
The upon block at \cref{ln:received-prepare} is executed the first time that all of the following conditions are met by validator $v$:
\begin{itemize}
	\item $v$ has accepted a Proposal message for the proposed block \pb{}, i.e. ${acceptedPB}_{h,v} = \pb{}$;
	\item $v$ has received, from non-proposing validators for the current round, at least $ \quorum{\nhmacro[]}-1 $ Prepare messages for the current instance of the \ibfpmthree{}, current round and with digest $H$ corresponding to the Keccak hash of the accepted proposed block $\pb{}$.
\end{itemize}

When all of the conditions listed above are met for the first time, then:
\begin{itemize}
	\item $v$ multicasts a Commit message \tsmessage{COMMIT}{h,r_{h,v},H,\cs{\pb{},v}}{v} to all validators (including itself), where \cs{\pb{},v}, called \emph{commit seal}, corresponds to the signature of $v$ over the proposed block $\pb{}$ and the current round number;
	\item $v$ sets ${latestPreparedProposedBlock}_{h,v}$ to the proposed block $\pb{}$;
	\item $v$ sets ${latestPC}_{h,v}$ to a set including the signed portion of the accepted Proposal message and all of the Prepare messages sent by non-proposing validators for the current round targeting the current instance of the \ibfpmthree{}, current round and with digest $H$ corresponding to the Keccak hash of the accepted proposed block $\pb{}$.
\end{itemize}
The pseudocode uses the state variable ${commitSent}_{h,v}$ to indicate that the Commit message is sent only the first time that all of the conditions listed above are met. Indeed, ${commitSent}_{h,v}$ is set to $true$ at \cref{ln:set-commit-sent} and reset to $false$ in the ${StartNewRound}$ procedure at \cref{ln:reset-commit-sent}.
By borrowing from the PBFT terminology, when a validator meets the conditions indicated in the \textbf{upon} block at \cref{ln:received-prepare}, then $v$ is said to be \nprepared{} at round $r_{h,v}$.
By borrowing again from the PBFT terminology, ${latestPC}_{h,v}$ is called latest Prepared-Certificate and the protocol is designed so that ${latestPC}_{h,v}$ always holds at least the minimum number of messages required to prove that $v$ \nprepared{} in round $r_{h,v}$ on a proposed block with Keccak hash $H$.
We say that a Prepared-Certificate is for round $r$ and instance $h$ if and only if the Prepared-Certificate only includes signed Proposal and Prepare messages for round $r$ and height $h$.
${latestPC}_{h,v}$ always holds the Prepared-Certificate for latest round in the current instance $h$ where $v$ is \nprepared{}.
${latestPC}_{h,v} = \bot$ only if $v$ has never \nprepared{} in the current instance $h$.

\vspace{1em}
The upon block at \cref{ln:received-commit} is executed the first time that all of the following conditions are met by validator $v$:
\begin{itemize}
	\item $v$ has accepted a Proposal message for the proposed block \pb{}, i.e. ${acceptedPB}_{h,v} = \pb{}$;
	\item $v$ has received, from at least different $ \quorum{\nhmacro[]}$ validators for the current instance of the \ibfpmthree{}, a Commit message for the current instance of the \ibfpmthree{}, current round, with digest $H$ corresponding to the Keccak hash of the accepted proposed block $\pb{}$ and commit seal signed by the sender of the Commit message.
\end{itemize}
When all of the conditions listed above are met for the first time, then:
\begin{itemize}
	\item $v$ creates a block finalisation proof modelled as a tuple comprising the current round number and the commit seals included in all the Commit messages that satisfy the condition on Commit messages stated above;
	\item $v$ creates a finalised block modelled as a tuple comprising the Ethereum block included in the proposed block and the finalisation proof;
	\item $v$ broadcasts the finalised block to all nodes;
\end{itemize}
The pseudocode uses the state variable ${finalisedBlockSent}_{h,v} $ to trigger the transmission of a finalised block only the first time that the conditions listed above are met.
${finalisedBlockSent}_{h,v} $ is set at \cref{ln:set-finalised-block-sent}  and reset in the ${StartNewRound}$ procedure at \cref{ln:reset-finalised-block-sent}.

In alignment with PBFT, IBFT 2.0 relies on a round change sub-protocol to detect whether the selected proposer may be Byzantine and causing the protocol to never terminate. 
As specified at \crefrange{ln:start-timer-begin}{ln:start-timer-end}, whenever a validator $v$ starts a new round, it starts a round timer with duration exponential to the round number (see \cref{ln:calcualate-round-timer-duration}).

When validator $v$'s round timer for the current round expires (\cref{ln:timer-expiry}), $v$ starts the round $r_{h,v}+1$ and multicasts a \tsmessageinnnsmessage{ROUND-CHANGE}{h, r_{h,v}+1, {latestPC}_{h,v}}{{latestPreparedProposedBlock}_{h,v}} message for the new round to all validators (including itself).
As it can be noted, the Round-Change message includes the latest Prepared-Certificate and the proposed block associated with the latest Prepared Certificate.

The \textbf{upon} block at \cref{ln:received-round-change} describes under which conditions a Proposal message for a new round is multicast and how the Proposal message is assembled.
Specifically, the \textbf{upon} block is executed only if $v$ has received at least \quorum{\nhmacro[]} Round-Change message for the current \ibfpmthree{} instance such that:
\begin{itemize}
	\item all messages are for the same round number $r_{rc}$ and $r_{rc}$ is higher than the current round;
	\item all messages are sent by distinct validators for the current instance of the \ibfpmthree{};
	\item all messages contain a valid Prepared-Certificate;\\
	A Prepared-Certificate is considered valid either if it is empty ($\bot$) or if it contains one Proposal message and at least $\quorum{\nhmacro[]} -1$ Prepare messages for the same round $r'$ such that:
	\begin{itemize}
		\item $r'<$ than the the round number of the Round-Change messages, $r_{rc}$;
		\item the Proposal message is signed by the selected proposer for round $r'$;
		\item the  Prepare messages are sent by non-proposing validators for the $h$-th instance of the \ibfpmthree{} and round $r'$ and they are all for the current instance $h$ of the \ibfpmthree{}, round $r'$ and same hash as the one included in the Proposal message.
	\end{itemize}
	\item the Keccak hash of the ${proposedBlock}$ included in each of the Round-Change messages considered matches the hash included in the Proposal and Prepare message of the Prepared-Certificate.
\end{itemize}
We say that any set meeting these conditions is a valid Round-Change-Certificate for round $r'$ where round $r'$ is the round included in all of the Round-Change messages included in the Round-Change-Certificate.
When all of the conditions listed above are met, then: 
\begin{itemize}
	\item $v$ moves to the round number included in one of the Round-Change-Certificates with the highest round number, let $RCC$ be the chosen Round-Change-Certificate and let  $r_h$ be the round number of $RCC$;
	\item if $v$ was not already in $r_h$, then $v$ starts the round timer for round $r_h$;
	\item if $v$ is the selected proposer for round $r_h$, then $v$ sends a 
	Proposal message for the new round including the selected Round-Change-Certificate ($RCC$) and a proposed block calculated as follows:
	\begin{itemize}
		\item if all of the Prepared-Certificates included in the Round-Change message are empty, i.e $= \bot$, then the proposed block must be a tuple including any valid Ethereum block for height $h$ and the current round number, which must match the round number of the Proposal message;
		\item otherwise, the Ethereum block included in the tuple constituting the proposed block must match the proposed block received as part of one of the Round-Change messages including a Prepared-Certificate for the highest round amongst the rounds of all the other Prepared-Certificates included in the Round-Change-Certificate.
	\end{itemize}
\end{itemize}

As specified by the \textbf{upon} block at \cref{ln:received-proposal-round-higher-than-0}, a Proposal message \tsmessageinnsmessage{PROPOSAL}{h, r_{pp}, H}{sender}{\pb{},RCC} for a round higher than 0 
is accepted only if all of the following conditions are met:
\begin{itemize}
	\item $v$ is currently running the \ibfpmthree{} instance   $ h_{pp} $, i.e $h_{pp} = h$ ;
	\item the round number $r_{pp}$ of the Proposal message is either higher than the current round or equal to the current round provided that $v$ has not accepted any Proposal message for the current round (i.e ${acceptedPB} = \bot$;);
	\item $H$ corresponds to the Keccak hash of the proposed block \pb{}.
	\item the signed portion of the message is signed by the selected proposer for round $r_{pp}$ and instance $h$ of the \ibfpmthree{};
	\item the Round-Change-Certificate $RCC$ includes at least \quorum{\nhmacro[]} Round-Change messages for round $r_{pp}$ and height $h$;
	\item if each of the Round-Change messages included in the Round-Change-Certificate $RCC$ includes either an invalid Prepared-Certificate or an empty Prepared-Certificate, then 
	\begin{itemize}
		\item the Ethereum block $\eb$ included in the proposed block \pb{} must be a valid block for height $h$;
		\item the round number $r_\eb$ included in the prepared block \pb{} must match the current round number, i.e. $r_{h,v} = r_{\eb}$;
	\end{itemize}
	\item otherwise, the Keccak hash of the tuple composed of the Ethereum block included in \pb{} and the highest round number of all Prepared-Certificated included in $RCC$ must match the hash included in any of the messages that are part of the Prepared-Certificates with the highest round number amongst all of the Prepared-Certificates included in the Round-Change-Certificate $RCC$.
\end{itemize}
The effect of accepting a Proposal message for a round $r_{pp}>0$ is essentially the same effect of accepting the Proposal message for round 0 with the addition of moving to round $r_{pp}$ , namely:
\begin{itemize}
	\item $v$ moves to round $r_{pp}$ and start related round timer if $v$ was not already in round $r_{pp}$;
	\item $v$ multicasts a Prepare message \tsmessage{PREPARE}{h, r_{h,v}, H}{v} to all validators (including itself);
	\item $v$ sets ${acceptedPB}_{h,v}$ to the proposed block \pb{} indicating that it accepted a Proposal message for \pb{}.
\end{itemize}

From here on the protocol proceeds as described above.

\afterpage{
\clearpage
\KOMAoptions{paper=A4,pagesize,DIV=50}
\recalctypearea
\begin{algorithm}[H]
	\caption{$h$-th instance of the \ibfpmthree{} for validator $v$.}
	\label{algo:ibpfpmone}
	\setcounter{AlgoLine}{0}
	\Expansions{
		\lExpand{\nhmacro[]}{\nmacro[]}\label{ln:nhmacro}
		\lExpand{\avhmacro[]}{\avibfp[]}\label{ln:avibfp}
		\lExpand{\hproposer[] } {\proposermacro[]}}\label{ln:hproposer}
	\BlankLine
	\Functions{ 
		\lFnInline{\cs{\pb{},v}}{\Sign{$ \kec{\pb{} }, v_{privKey}$}}
		\BlankLine
		\lFnInline{${validPC}(PC,r_{limit},h,v)$}
		{
			$\left \lbrace
			\begin{tabular}{lp{12cm}}			
			\True & if $PC = \bot \Or \newline
				(\newline
				\mbox{\quad}\text{\tcc{$PC$ contains at least \quorum{\nhmacro[]} messages}}\newline
				\mbox{\quad} \SizeOf(PC) \geq \quorum{\nhmacro[]} \And 	\vspace{0.3em} \newline	
				\mbox{\quad}\text{\tcc{$PC$ contains one and only one Proposal message}}\newline
				\mbox{\mbox{\quad}\SizeOf{\ensuremath{\{ m \in PC \With m = \tsmessage{PROPOSAL}{h,*,*}{*} \}}}=1 \And} \vspace{0.3em} \newline
				\mbox{\quad}\text{\tcc{the other messages in $ PC  $ are all Prepare messages}}\newline
				\mbox{\mbox{\quad}\SizeOf{\ensuremath{\{ m \in PC \With m = \tsmessage{PREPARE}{h,*,*}{*} \}}}$ = \SizeOf{PC} -1 $ }  \vspace{0.3em}\newline	
				\mbox{\quad}\text{\tcc{all the messages in $ PC $ are sent by different senders}}\newline				
				\mbox{\quad}\mbox{\ensuremath{\Any \smessage{m'}{sender'}, \smessage{m''}{sender''} \In PC \IsSuchThat m' \neq m'' \:\Implies\: sender' \neq sender'' \And }}\vspace{0.3em}\newline				
				\mbox{\quad}\text{\tcc{all the messages in $PC$ are for the same round and proposed block hash}}\newline				
				\mbox{\quad}\mbox{\ensuremath{\Any \tsmessage{*}{h,r',H'}{*}, \tsmessage{*}{h,r'',H''}{*} \In PC \IsSuchThat r' = r'' \And H'=H'' \And }	}\vspace{0.3em}\newline		
				\mbox{\quad}\text{\tcc{the round included in all the messages in $PC$ is lower than $r_{limit}$}}\newline				
				\mbox{\quad}\mbox{\ensuremath{\Any \tsmessage{*}{h,r',*}{*}\In PC \IsSuchThat r' < r_{limit} \And  }	}\newline
				\mbox{\quad}
				\begin{minipage}{15cm}
					{\tcc{Proposal messages in $PC$ are sent by the proposer for the round included in the Proposal message, whereas Prepare messages are sent by any non-proposing validator for the round included in the Prepare messsages}}
				\end{minipage}
				\newline												
				\mbox{\quad}(\newline
				\mbox{\qquad}\Any m \In PC \IsSuchThat \newline
				\mbox{\qquad\quad}m =\tsmessage{PROPOSAL}{h,\fv{r'},*}{sender} \And sender = \Proposer{\ensuremath{\fv{r'}}} \Or \newline
				\mbox{\qquad\quad}m =\tsmessage{PREPARE}{h,\fv{r'},*}{sender} \And sender \In (\avhmacro[] \,\NotIncluding\, \Proposer{\ensuremath{\fv{r'}}})\newline
				\mbox{\quad})\newline
				)$\\
			\False & otherwise
			\end{tabular}\right.$
		}
		\lFnInline{$ {RoundTimerTimeout}(r)$}{$\TimeoutForRoundZero \cdot 2^r$ \label{ln:calcualate-round-timer-duration}}		
	}
	\BlankLine	
	\Init{\label{ln:initialisation}
		${latestPC}_{h,v} \gets \bot $\;
		${latestPreparedProposedBlock}_{h,v} \gets \bot $\;
		${finalisedBlockSent}_{h,v} \gets {false}$\label{ln:reset-finalised-block-sent}\;
		$ StartNewRound(0,h,v) $\;
		\If{$v = \Proposer{\ensuremath{0}} $}{
			$\Let\: \pb{} \equiv (\CreateNewProposedBlock{\ensuremath{h,v}},0)$\;
			${acceptedPB}_{h,v} \gets \pb{}$\label{ln:set-accepted-proposal}\;	
			\Multicast{\tsmessageinnsmessage{PROPOSAL}{h, r_{h,v}, \kec{\pb{}}}{v}{\pb{},\bot}} to \avhmacro[]\label{ln:broadcast-prepare}\;	
		}		
	}
	\BlankLine	
	\Procedures{
		\Proc{$ StartNewRound(r,h,v) $}{
			\If{$r = 0 \Or r > r_{h,v}$\label{ln:start-timer-begin}}
			{
				\Set ${roundTimer}_{h,v}[r]$ to expire after ${RoundTimerTimeout}({\ensuremath{r}})$\;
				\label{ln:start-timer-end}
			}
			$r_{h,v} \gets r$\;
			${acceptedProposedPB}_{h,v} \gets \bot$\label{algo3:set-accept-pre-prepare-1}\;
			${commitSent}_{h,v} \gets {false}$\label{ln:reset-commit-sent}\;
		}
		\BlankLine
		
		\Proc{$StartNewRoundAndSendRoundChange(r, h, v)$}
		{
			$StartNewRound(r, h, v)$\;
			\tcc{$r_{h,v} = r $ from this point on}
			\BlankLine
			\Multicast{\tsmessageinnsmessage{ROUND-CHANGE}{h, r_{h,v}, {latestPC}_{h,v}}{v}{{latestPreparedProposedBlock}_{h,v}}} to \avhmacro[]\;			 	
		}
	}	
\end{algorithm}
\pagebreak
\addtocounter{algocf}{-1}
\begin{algorithm}[H]
	\caption{$h$-th instance of the \ibfpmthree{} for validator $v$ (continue).}
	\Indp		
	
	\UponRules{
	}
	\Indp
		\tcc{reception of PROPOSAL messages for round 0}
		\Upon{\label{ln:upon-pre-prepare}\color{\uponcolor}\ensuremath{{acceptedPB}_{h,v} = \bot \And \newline
					\mbox{\hspace{0.8em}}(\ThereExists \tsmessageinnsmessage{PROPOSAL}{h, \fv{r_{p}}, \fv{H}}{\fv{sender}}{\fv{\pb{}}:(\fv{EB},\fv{r_{EB}}),\bot} \In {receivedMessages}_v \With: \newline
			 		\mbox{\qquad}\fv{sender} = \hproposer[] \And \newline 
				 	\mbox{\qquad}\ensuremath{\fv{H} = \kec{\fv{\pb{}}} \And \newline
				 	\mbox{\qquad}\fv{r_{p}}=r_{h,v} = \fv{r_{EB}} = 0 \And \newline
				 	\mbox{\qquad}\Valid{\ensuremath{\fv{\eb},{{chain_{EB}}_v}[h-1]}}\newline
				 	\mbox{\hspace{1em}})}   \And \newline
					\mbox{\hspace{0.8em}} v  \neq \hproposerp{ 0 } }
			 }{
					${acceptedPB}_{h,v} \gets \pb{}$\label{ln:upon-pre-prepare:end-accept}\;
					\Multicast{\tsmessage{PREPARE}{h, r_{h,v}, \fv{H}}{v}} to \avhmacro[]\;				
		}
		\BlankLine		
		\tcc{reception of PREPARE messages}
		\Fn{$ validPrepareMessages(h,v) $}
		{
			\mbox{\hspace{0.8em}}\ensuremath{\{\tsmessage{PREPARE}{h, \fv{r_{c}}, \fv{H}}{\fv{sender}} \In {receivedMessages}_v \With \newline
				\mbox{\qquad}\fv{sender} \In  \avhmacro[] \And \newline 
				\mbox{\qquad}\fv{r_c} = r_{h,v} \And \newline
				\mbox{\qquad}\fv{H} = \kec{{acceptedPB}_{h,v}} 
				\} }
		}	
		\BlankLine			
		\Upon{\color{\uponcolor}\label{ln:received-prepare}${acceptedPB}_{h,v} \neq \bot \And \newline 
			   \mbox{\hspace{0.8em}}$\SizeOf{$ validPrepareMessages(h,v) $}  $\geq \quorum{\nhmacro[]}-1  \And \newline 
		       \mbox{\hspace{0.8em}}{commitSent}_{h,v} = {false}$ }{

				$\Let\: \pb{} \equiv {acceptedPB}_{h,v} $\;
				\Multicast{\tsmessage{COMMIT}{h, r_{h,v}, \kec{\pb{}},\cs{\pb{},v}}{v}} to \avhmacro[]\label{ln:algo3:send-commit}\;
				${commitSent}_{h,v} \gets {true}$\label{ln:set-commit-sent}\;	
				\tcc{$ latestPC_{h,v} $ is set to a set including the signed portion of the accepted Proposal message and all the valid Prepare messages received}		
				${latestPC}_{h,v} \gets $
				$\{  \tsmessage{PROPOSAL}{h,r_{h,v},\kec{\pb{}}}{\hproposer{}}\With \newline
				\mbox{\hspace{8em}}\ThereExists \mmessage{ \tsmessage{PROPOSAL}{h,r_{h,v},\kec{\pb{}}}{\hproposer{}}, *,*} \In receivedMessages_v  \} \newline
				\mbox{\hspace{5em}}\textbf{union with }  validPrepareMessages(h,v) $
				\;
				${latestPreparedProposedBlock}_{h,v} \gets \pb{} $\;				

		}	
	\BlankLine
	\tcc{reception of COMMIT messages}
	\Fn{$ validCommitMessages(h,v) $}
	{
		\mbox{\hspace{0.8em}}\ensuremath{\{\tsmessage{COMMIT}{h, \fv{r_{c}}, \fv{H}, \fv{cs}}{\fv{sender}} \In {receivedMessages}_v \With \newline
		\mbox{\qquad}\fv{sender} \In  \avhmacro[] \And \newline 
		\mbox{\qquad}\fv{r_c} = r_{h,v} \And \newline
		\mbox{\qquad}\fv{H} = \kec{{acceptedPB}_{h,v}} \And \newline
		\mbox{\qquad}\EcRecover{\ensuremath{\kec{{acceptedPB}_{h,v}},\fv{cs}}} = \fv{sender} \} }
	}
	\BlankLine
	\Upon{\label{ln:received-commit}\color{\uponcolor}${acceptedPB}_{h,v} \neq \bot \And \newline $
		\mbox{\hspace{0.8em}}\SizeOf{$ validCommitMessages(h,v) $} $\geq \quorum{\nhmacro[]} \And \newline
		{finalisedBlockSent}_{h,v} = {false}  $ 
		\label{ln:new-send-finalisation-proof-upon-condition}
		}{
			$\Let {commitSeals} \equiv \newline
			\{ cs \Withc \newline
			\mbox{\qquad\quad}\ThereExists\tsmessage{COMMIT}{h, *,*,cs}{*} \In validCommitMessages(h,v) 
			\}$\;
			$\Let\:{finalisationProof}_{h,v} \equiv (r_{h,v},commitSeals)$\; 
			$\Let\: \eb \equiv \pi_1(acceptedPB_{h,v})$\;
			$\Let \:{\fb} \equiv (\eb , {finalisationProof}_{h,v}) $\;
			\Broadcast{$ \tfakensmessage{FINALISED-BLOCK}{\fb} $} to all nodes\;
			${finalisedBlockSent}_{h,v} \gets {true}$\label{ln:set-finalised-block-sent}\;

	}
\end{algorithm}
\pagebreak
\addtocounter{algocf}{-1}
\begin{algorithm}[H]
	\caption{$h$-th instance of the \ibfpmthree{} for validator $v$ (continue).}
	\fontsize{9}{10}\selectfont
	\Indp	
	\tcc{round timer expiry}
	\Upon{\color{\uponcolor}\Expiry ${roundTimer}_{h,v}[r_{h,v}]$\label{ln:timer-expiry}} {
		$StartNewRoundAndSendRoundChange(r_{h,v}+1, h, v)$\;
	}
	\BlankLine	
	\tcc{reception of ROUND-CHANGE messages}
	\Fn{${setsOfValidRoundChangeCertificates}(h,v)$}
	{
		\ensuremath{\{RCC \In \AllSubsetsOf{\ensuremath{{receivedMessages}_v}} \Withc \newline	
			\SizeOf(RCC) = \quorum{\nhmacro[]} \And \vspace{0.3em}\newline
			\text{\tcc{all messages in $ RCC $ are Round-Change messages for instance $ h $}}\newline
			\Any m_{RCC} \In RCC \IsSuchThat m_{RCC} = \tsmessageinnsmessage{ROUND-CHANGE}{h,*, *}{*}{*}  \And \vspace{0.3em} \newline			
			\text{\tcc{all messages in $ RCC $ are sent by different senders}}\newline
			\Any \mmessage{\smessage{m_1}{sender},*},\mmessage{\smessage{m_2}{sender'}, *} \In RCC \IsSuchThat m_1 \neq m_2 \:\Implies sender \neq sender' \And \vspace{0.3em} \newline
			\text{\tcc{all messages in $ RCC $ are for the same round}}\newline
			\mbox{\ensuremath{\Any \tsmessageinnsmessage{ROUND-CHANGE}{*, r', *}{*}{*} ,\tsmessageinnsmessage{ROUND-CHANGE}{*, r'', *}{*}{*}  \In RCC \IsSuchThat r' = r'' \And }}\vspace{0.3em} \newline				
			\text{\tcc{any message in $RCC$ is such that:}}\newline
			\Any \tsmessageinnsmessage{ROUND-CHANGE}{h, \fv{r_{rc}}, \fv{PC}}{\fv{sender}}{\fv{\pb{}}}  \In RCC \IsSuchThatc \newline  
			{\qquad} \text{\tcc{the sender is one of the validators}} \newline
			{\qquad} \ensuremath{\fv{sender} \In \AV_{h} \And \vspace{0.3em} \newline
				{\qquad} 
				\begin{minipage}{16cm}
					\tcc{the round number of the Round-Change message is either higher than the current round number or equal to the current round number provided that no Proposal message for the round included in the Round-Change message has alrady been accepted}
				\end{minipage}
				\newline
				{\qquad}	(\fv{r_{rc}} > r_{h,v} \Or (\fv{r_{rc}} = r_{h,v} \And {acceptedPB} = \nil{})) \And \vspace{0.3em} \newline}  
			{\qquad} \text{\tcc{the Prepared-Certificate is valid}} \newline
			{\qquad}  \ensuremath{{validPC}(\fv{PC},\fv{r_{rc}},h,v) \And \vspace{0.3em} \newline
				{\qquad} 
				\begin{minipage}{14cm}		    	
					\tcc{the block hash included in all of the messages included in the Prepared certificate corresponds to the hash of the proposed block included in the Round-Change message}
				\end{minipage}
				\newline	
				{\qquad} \Any m_{PC} \In \fv{PC} \IsSuchThatc \newline
				{\qquad\quad} m_{PC} = \tnsmessage{PROPOSAL}{*,*,\kec{\fv{\pb{}}}} \Or m_{PC} = \tnsmessage{PREPARE}{*,*,\kec{\fv{\pb{}}}}
			}\}}\newline
	}
	\BlankLine	
	\Upon{\color{\uponcolor}\SizeOf{${setsOfValidRoundChangeCertificates}(h,v)$} $ \geq 1  $\label{ln:received-round-change}}{
				\BlankLine
				\tcc{\Let ${extendedRCC}$ be any valid Round-Change-Certificate for the highest round number}
				\Let ${extendedRCC} \equiv \newline \ChooseSet \newline
				\{eRCC \In {setsOfValidRoundChangeCertificates}(h,v) \Withc \newline
				{\quad}\Any \tsmessageinnsmessage{ROUND-CHANGE}{*, \fv{r}, *}{*}{*} \In eRCC \IsSuchThatc:\newline
				{\qquad} \fv{r} = \max(\{ r' \With \ThereExists \tsmessageinnsmessage{ROUND-CHANGE}{*, r', *}{*}{*} \In extendedRCC \}) $\; 
				\BlankLine
				\tcc{$\Let\: r_{rc} $ be the round number of the Round-Change certficate in ${extendedRCC}$ }
				$\Let\:r_{rc} \equiv \Choose \{r \With  \tsmessageinnsmessage{ROUND-CHANGE}{*, r, *}{*}{*} \In extendedRCC\} $\;
				$StartNewRound(r_{rc},h,v)$\;	
				\tcc{$r_{h,v} = r_{rc}$ from this point on}				
				\BlankLine							
				\If{$ v = \Proposer{\ensuremath{r_{rc}}} $}
				{
					\tcc{\Let $roundsAndPreparedBlocks$ be a set of all tuples where the the first element is the proposed block included in the Round-Change messages in $ extendedRCC $ including a non-empty Prepared-Certificate and the second element is the round number of those Prepared-Certificates}
					\Let ${roundsAndPreparedBlocks} \equiv  \{(\eb{},r) \Withc \newline
					\mbox{\hspace{12em}}\ThereExists\tsmessageinnsmessage{ROUND-CHANGE}{*, *, \fv{PC}}{*}{\fv{\pb{}}} \In {extendedRCC} \Withc \newline 
					\mbox{\hspace{12em}\quad}\fv{\pb{}} \neq \bot \And  \fv{\pb{}} = (\eb,*) \And \newline
					\mbox{\hspace{12em}\quad} \fv{PC }\neq \bot \And  \Any m \In \fv{PC} \IsSuchThat  m = \tnsmessage{*}{*,r,*}
					\}$\;		
					\BlankLine
					\eIf{${roundsAndPreparedBlocks} = \{\}$}
					{
						\tcc{If all of Prepared-Certificates included in the Round-Change messages included in $ extendedRCC $ are empty, then the Ethereum block included in the proposed block  must be any valid Ethereum block for height $h$}
						$\Let\: \eb{} \equiv \CreateNewProposedBlock{\ensuremath{h,v}}$ \;
					}
					{
						\tcc{If at least one of the Prepared-Certificates included in the Round-Change messages included in $ extendedRCC $ is not empty, then the Ethereum block included in the proposed block must match the Ethereum block included in the Prepared-Certificate with the highest round number }
						$\Let\: \eb{} \equiv \Choose \newline
						\{ \eb{} \With  \newline
						\mbox{\qquad}(\eb{},r) \In {roundsAndPreparedBlocks} \And \newline
						\mbox{\qquad}r = \max(\{ r' \With (*,r') \In roundsAndPreparedBlocks \}) \}$\;
					}
					\tcc{The proposed block is modelled as a tuple where the first element is the proposed Ethereum block and the second element is the current round}
					$\Let\: \pb{}\equiv (\eb,r_{h,v})$\;
					\BlankLine
					\tcc{$\Let\:RCC$ be the signed portion of the Round-Change messages included in ${extendedRCC}$ }
					$\Let\: \mathit{RCC} \equiv \{\smessage{RC}{sender} \With \mmessage{\smessage{RC}{sender},*} \In {extendedRCC}  \}$\;
					\Multicast{\tsmessageinnsmessage{PROPOSAL}{h, r_{h,v}, \kec{\pb{}}}{v}{\pb{},\mathit{RCC}}} to \avhmacro[]\; 
					${acceptedPB}_{h,v} \gets \pb{}$\;					
				}
	}
\end{algorithm}
\pagebreak
\addtocounter{algocf}{-1}
\begin{algorithm}[H]
	\caption{$h$-th instance of the \ibfpmthree{} for validator $v$ (continue).}
	\Indp
			\tcc{reception of PROPOSAL messages for rounds higher than 0}
			\Upon{\label{ln:received-proposal-round-higher-than-0}\color{\uponcolor}{\ensuremath{\ThereExists\tsmessageinnsmessage{PROPOSAL}{h, \fv{r_{pp}}, \fv{H}}{\fv{sender}}{\fv{\pb{}}:(\fv{EB},\fv{r_{EB}}),\fv{\mathit{RCC} }} \In {receivedMessages}_v \Withc \newline
						\mbox{\qquad}\fv{sender}= \Proposer{\ensuremath{\fv{r_{pp}}}} \And \newline
						\mbox{\qquad}\fv{H} = \kec{\fv{\pb{}}} \And \vspace{0.3em} \newline 
						\mbox{\qquad}
					 	\begin{minipage}{16cm}
					 		\tcc{the round number of the Proposal message is either higher than the current round number or equal to the current round number provided that no Proposal message for the round included in the Round-Change message has alrady been accepted}
					 	\end{minipage}		
					 	\newline				
						\mbox{\qquad}(\fv{r_{pp}}>r_{h,v} \Or (\fv{r_{pp}} = r_{h,v} \And {acceptedPB}_{h,v} = \bot)) \And \newline 
						\mbox{\qquad}\text{\tcc{the Round-Change-Certificate contains at least \quorum{\nhmacro[]} messages}}\newline
						\mbox{\qquad}\SizeOf{RCC} \geq \quorum{\nhmacro[]} \And \vspace{0.3em}\newline
						\mbox{\qquad}
						\begin{minipage}{16cm}
							{\tcc{any message in $RCC$ is a Round-Change message for height $h$ and round equal to the round of the Proposal message sent by a validator}}
						\end{minipage}
						\newline
						\mbox{\qquad}\Any m \In \fv{RCC} \IsSuchThat \newline
						\mbox{\qquad\quad}m = \tsmessage{ROUND-CHANGE}{h, \fv{r_{pp}},*}{\fv{sender'}} \And \newline
						\mbox{\qquad\quad}sender' \In \avhmacro[] \And \vspace{0.3em}	\newline				
						\mbox{\qquad}
						\begin{minipage}{16cm}
							{\tcc{any message in $RCC$ is signed by a different validator}}
						\end{minipage}
						\newline							
						\mbox{\qquad}\Any \smessage{m_1}{sender'},\smessage{m_2}{sender''} \In RCC \IsSuchThat m_1 \neq m_2 \:\Implies\: sender' \neq sender'' \And
						}\newline
						 $\mbox{\hspace{0.8em}}v \neq \Proposer{\ensuremath{\fv{r_{pp}}}} $  }}{
				
				\BlankLine
				\tcc{\Let $roundsAndPreparedBlockHashes$ be the set of all tuples where the the first element is the block hash included in  non-empty and valid Prepared-Certificates included in each of the Round-Change messages in $ RCC $ and the second element is the round number of those Prepared-Certificates  }
				\Let ${roundsAndPreparedBlockHashes} \equiv  \{( H,r) \Withc \newline
				\mbox{\quad}\ThereExists \tsmessage{ROUND-CHANGE}{h, \fv{r_{pp}},\fv{PC}}{*} \In RCC \Withc  \newline
				\mbox{\hspace{10.6em}\quad}\fv{PC} \neq \bot$ \And \newline
				$\mbox{\hspace{10.6em}\quad}{validPC}(\fv{PC},\fv{r_{pp}},h,v) \And \newline
				\mbox{\hspace{10.6em}\quad}\Any m \In \fv{PC} \IsSuchThat  m = \tnsmessage{*}{*,r,H}   \}$\;
				\BlankLine
				\If{${roundsAndPreparedBlockHashes} \neq \{\}$}
				{
					\tcc{If at least one of the Prepared-Certificates included in the Round-Change messages in $RCC$ is non-empty and valid then \:\Let $ maxR $ be the round number of the Prepared-Certificates included in the Round-Change messsages in $ RCC $ with the highest round number and \:\Let $ exptectedH $ be the block hash included in the  Prepared-Certificates with the highest round number}
					$\Let\: maxR \equiv \max(\{ r' \With (*,r') \In roundsAndPreparedBlockHashes \})$\;
					$\Let\: expectedH \equiv \Choose \newline
					\{ H \With  \newline
					\mbox{\qquad}(H,r) \In {roundsAndPreparedBlockHashes} \And \newline
					\mbox{\qquad}r = maxR \}$\;
				}
				\If{\tcc{there is no a valid and non-empty Prepared-Certificate included in the Round-Change messages in $RCC$ and the Ethereum block included in the proposed block is a valid Ethereum block for height $h$}
					 \mbox{\quad}$({roundsAndPreparedBlockHashes} = \{\} \And \Valid{\ensuremath{\fv{\eb},{{chain_{EB}}_v}[h-1]}} )$ \Or \vspace{0.3em} \newline
					 \tcc{the hash of the tuple composed of the Ethereum block included in the proposed block and the highest round number of the Prepapred-Certificates included in the Round-Change messages in $RCC$ matches the expected hash}
					 \mbox{\quad}$({roundsAndPreparedBlockHashes}\neq \{\} \And \kec{(\fv{\eb{}},maxR)} = {expectedH})$}
				{
					\vspace{0.3em}
					$StartNewRound(\fv{r_{pp}}, h, v)$\;
					\tcc{$r_{h,v} = r_{pp}  $ from this point on}
					\BlankLine
					${acceptedPB}_{h,v} \gets \fv{\pb{}}$\label{algo3:set-accept-pre-prepare-2}\;
					\Multicast{\tsmessage{PREPARE}{h, r_{h,v}, \fv{H}}{v}} to \avhmacro[]\label{ln:send-prepare-on-reception-of-proposal-for-round-higher-than-zero}\;
				}
			}
\end{algorithm}
\clearpage

}

\pagebreak
\KOMAoptions{paper=A4,pagesize,DIV=15}
\recalctypearea
\recalctypearea

\section{Robustness Analysis}\label{sec:robustness-analysis}

In this section we show that the IBFT 2.0 protocol is robust when operating in an eventually synchronous network.\\
As such the eventually synchronous network model assumption is assumed throughout this analysis.

\subsection{Definitions}
In this section we provide a few definitions that will be used in the following sections to draw the robustness analysis of the IBFT 2.0 protocol to conclusion.
Most of these definitions were first introduced in \cite{IBFT1-Analysis} and are re-stated here for completeness. 

$t$-Byzantine-fault-tolerant Persistence is defined as follows.
\begin{definition}[$t$-Byzantine-fault-tolerant Persistence.]
	The IBFT 2.0 protocol ensures $t$-Byzantine-fault-tolerant Persistence if and only if the following statement is true:
	provided that no more than $t$ validators are Byzantine, the IBFT protocol guarantees the Persistence property defined in \cref{def:robustnes}. 
\end{definition}

In the context of the \ibfpmthree{}, Safety is defined as follows:
\begin{definition}[$t$-Byzantine-fault-tolerant Safety for the \ibfpmthree{}]\label{def:t-tolerate-ibfpmthree}
	The \ibfpmthree{} ensures $t$-Byzantine-fault-tolerant Safety if and only if it guarantees the validity of the following statement:
	in the presence of no more than $t$ Byzantine validators, the protocol ensures that no two valid finalised blocks including different Ethereum blocks for the same height can ever be produced.
\end{definition}

In relation to the Safety property of the \ibfpmthree{}, we define Byzantine-fault-tolerant Safety threshold as follows.
\begin{definition}[Byzantine-fault-tolerant Safety threshold]
	Byzantine-fault-tolerant Safety threshold for a protocol that guarantees $t$-Byzantine-fault-tolerant Safety is defined as the maximum number of Byzantine nodes that the protocol can withstand while ensuring Safety, i.e. $t$.
\end{definition}

The following two definitions are related to the Liveness property of the IBFT 2.0 protocol.
\begin{definition}[$t$-Byzantine-fault-tolerant Liveness]
	The IBFT 2.0 protocol ensures $t$-Byzantine-fault-tolerant Liveness if and only if the following statement is true:
	provided that no more than $t$ validators are Byzantine, the IBFT 2.0 protocol guarantees the Liveness property defined in \cref{def:robustnes}. 
\end{definition}

\begin{definition}[$t$-Byzantine-fault-tolerant Weak-Liveness of the \ibfpmthree{}]
	The \ibfpmthree{} guarantees $t$-Byzantine-fault-tolerant Weak-Liveness if and only if, provided that no more than $t$ validators are Byzantine, it guarantees that for any $h$ instance of the \ibfpmthree{} at least one valid finalised block for height $h$ will eventually be produced.
\end{definition}

In relation to the Weak-Liveness property of the \ibfpmthree{}, we define Byzantine-fault-tolerant Weak-Liveness threshold as follows: 
\begin{definition}[Byzantine-fault-tolerant Weak-Liveness threshold]
	Byzantine-fault-tolerant Weak-Liveness threshold for a protocol that guarantees $t$-Byzantine-fault-tolerant Weak-Liveness is defined as the maximum number of Byzantine nodes that the protocol can withstand while ensuring Weak-Liveness, i.e. $t$.
\end{definition}

As proved in \citewithauthor{Dwork:1988:CPP:42282.42283}, when the network communication is eventually synchronous, consensus is deterministically possible in a network of $n$ nodes participating in the consensus protocol if and only if no more than  $f(n) \equiv \left \lfloor \frac{n-1}{3} \right \rfloor $ of these nodes are Byzantine.
The following definitions of optimal Byzantine-fault-tolerant Safety, Persistence, Weak-Liveness and Liveness follow directly from this known lower limit.
\begin{definition}[Optimal Byzantine-fault-tolerant Safety threshold for the \ibfpmthree{} ]\label{def:optimal-byz-fault-tol-safety-for-ibfpmthree}
	The \ibfpmthree{} guarantees optimal Byzantine-fault-tolerant Safety threshold provided that for any instance $h$ its Byzantine-fault-tolerant Safety threshold corresponds to $f(n_h)$ where $n_h$ is the number of validators for the $h$-th instance of \ibfpmthree{}.
\end{definition}

\begin{definition}[Optimal Byzantine-fault-tolerant Persistence threshold for the IBFT protocol ]\label{def:optimal-Byzantine-fault-tolerance-persistence}
	The IBFT 2.0 protocol guarantees optimal Byzantine-fault-tolerant Persistence threshold if and only if
	it guarantees the Persistence property defined in \Cref{def:robustnes} despite up to $f(n_h)$ validators being Byzantine for any instance $h$ of the \ibfpmthree{}  where $n_h$ is the number of validators for the $h$-th instance of \ibfpmthree{}.
\end{definition}

\begin{definition}[Optimal Byzantine-fault-tolerant Weak-Liveness threshold for the \ibfpmthree{} ]\label{def:optimal-byz-fault-tol-weak-liveness-for-ibfpmthree}
	The \ibfpmthree{} guarantees optimal Byzantine-fault-tolerant Weak-Liveness threshold provided that for any instance $h$ its Byzantine-fault-tolerant Weak-Liveness threshold corresponds to $f(n_h)$ where $n_h$ is the number of validators for the $h$-th instance of \ibfpmthree{}.
\end{definition}

\begin{definition}[Optimal Byzantine-fault-tolerant Persistence threshold for the IBFT protocol ]\label{def:optimal-Byzantine-fault-tolerance-liveness}
	The IBFT 2.0 protocol guarantees optimal Byzantine-fault-tolerant Liveness threshold if and only if
	it guarantees the Liveness property defined in \Cref{def:robustnes} despite up to $f(n_h)$ validators being Byzantine for any instance $h$ of the \ibfpmthree{} where $n_h$ is the number of validators for the $h$-th instance of \ibfpmthree{}..
\end{definition}

\subsection{Safety Analysis of the \ibfpmthree{}}
In this section we prove that the \ibfpmthree{} provides optimal Byzantine-fault-tolerant Safety.

We use the inductive assumption used by Lemma 3 in \cite{IBFT1-Analysis}, which states that the local blockchains of all honest nodes are identical until finalised block with height $h-1$.
Therefore, since the set of validators for the $h$-th instance of the \ibfpmthree{} is a function of the local blockchain until the block with height $h-1$, this set is identical amongst all honest validators.
We denote the total number of validators for the $h$-th instance of the \ibfpmthree{} with $n_h$.\\
We also  assume that for any instance $h$ of the \ibfpmthree{}, no more than $f(n_h)$ validators are Byzantine.

\begin{lemma}\label{lem:intersection-of-q-and-}
	The intersection of any two sets of $ \quorum{n_h} $ validators includes an honest validator.
\end{lemma}
\begin{proof}
	See Lemma 23 of \cite{IBFT1-Analysis}.
\end{proof}

\begin{lemma}\label{lem:intersection-q-f-with-q}
	The intersection of any set of $ \quorum{n_h} - f(n_h) $ honest validators with any set of \quorum{n_h} validators includes an honest validator.
\end{lemma}
\begin{proof}
	Obvious as the intersection of any two sets of \quorum{n_h} validators includes an honest validator and any set of \quorum{n_h} validators is guaranteed to include at least $ \quorum{n_h} - f(n_h) $ honest validators.
\end{proof}

\begin{lemma}\label{lem:honst-validator-send-only-one-propose-or-prepare-message-for-each-round}
	An honest validator sends either one and only one Proposal or one and only one Prepare message for a round $r$.
\end{lemma}
\begin{proof}
	This is obvious from the pseudocode.
\end{proof}

\begin{corollary}\label{cor:an-honest-node-never-send-different-propose-or-prepare-message-for-the-same-round}
	An honest validator never sends two Proposal or Prepare messages for the same round $r$ with different block hashes.
\end{corollary}

\begin{lemma}
	For each round $r$, if two honest validators send a Commit message, then the block hashes included in the two Commit messages are the same. 
\end{lemma}
\begin{proof}
	Assume that the Lemma is false.
	This implies that two honest validators send two Commit messages including different block hashes.
	This, in turn, implies that each validator received at least \quorum{n_h} messages (one Proposal message and $ \quorum{n_h}-1 $ Prepare messages) for each of the two different block hashes.
	Since two sets of size \quorum{n_h} validators always intersect in an honest validator, this implies that one honest validator sent a Proposal or Prepare message for one of the block hashes and a Proposal or Prepare message for the other block hash.
	This is a contradiction as \Cref{lem:honst-validator-send-only-one-propose-or-prepare-message-for-each-round,cor:an-honest-node-never-send-different-propose-or-prepare-message-for-the-same-round} show that an honest validator never sends two different Prepare messages or two different Proposal messages and it sends either a Prepare message or a Proposal message for a given round, but never both.
\end{proof}

\begin{lemma}\label{lem:no-two-valid-pc-can-be-created}
	No two Prepared-Certificates for the same round and for different block hashes can be created.
\end{lemma}
\begin{proof}
	By, contradiction assume that the Lemma is false and two Prepared-Certificates, say $PC$ and $PC'$ are created for the same round $r$ but for different block hashes, say $H$ and $H'$ respectively.
	Since (i) a Prepared-Certificate includes at least \quorum{n_h} messages (between one Proposal message and $\quorum{n_h}-1$ Prepare messages) and (ii) according to 
	\cref{lem:intersection-of-q-and-} any two sets of \quorum{n_h} validators are guaranteed to intersect in at least one honest validator, this implies that at least one honest validator sent a Proposal or Prepare message for round $r$ and block hash $H$ and a Proposal or Prepare message for round $r$ and block hash $H'$.
	This is in contradiction with \Cref{lem:honst-validator-send-only-one-propose-or-prepare-message-for-each-round,cor:an-honest-node-never-send-different-propose-or-prepare-message-for-the-same-round}.	
\end{proof}

\begin{lemma}\label{lem:if-finalised-block-is-created-only-that-block-can-be-proposed-afterwards}
	If an honest validator $v$ creates a valid finalised block for the Ethereum block \eb{} while in round $r$, then the proposed block included in any valid
	Proposal message sent after round $r$ includes the same Ethereum block \eb{}. 
\end{lemma}
\begin{proof}

	If an honest validator $v$ creates a valid finalised block including the Ethereum block \eb{} while in round $r$, then at least $ \quorum{n_h} - f(n_h) $ honest validators must have sent a valid Commit message for round $r$ and the same block hash $H = \kec{(\eb,r)}$.
	This, in turn, implies that all of these $ \quorum{n_h} - f(n_h) $ honest validators have set their respective $latestPC$ variable to a Prepared-Certificate for block hash $H$ at round $r$.
	
	The proof is by induction on the round number.
	For the base case we show that the Lemma holds for any valid Proposal message sent for round $r+1$.
	Since any valid Proposal message includes \quorum{n_h} Round-Change messages, \Cref{lem:intersection-q-f-with-q} implies that at least one of these contains a Prepared-Certificate sent by an honest validator, say $v'$, that also sent a valid Commit message for block hash $H$ at round $r$.
	Since $v'$ sent a Commit message for block hash $H$ at round $r$, $v'$ must have included a valid Prepared-Certificate for round $r$ and block hash $H$ in its Round-Change message.
	Since valid Round-Change messages for round $r+1$ can only contain Prepared-Certificates for round number up to $r$, and according to \cref{lem:no-two-valid-pc-can-be-created} no two valid Prepared-Certificates for the same round can be created, a valid Proposal message for round $r+1$ must include a proposed block including an Ethereum block $ \eb' $ such that $ \kec{(r,\eb')} = H $.
	This together with our assumption on the uniqueness property of the Keccak function imply that $ \eb' = \eb $.
	

    For the inductive step, we assume that the Lemma is valid for $r'>r$, and then show that the Lemma is also valid for $r''=r'+1$.
    Consider that based on our inductive assumption, for any round $r'''$ such that  $r<r''' \leq r'$, if a valid Proposal message is sent then this message includes a proposed block including \eb{}.
    This implies that if any honest validator created a Prepared-Certificate while in any round $r'''$, then this Prepared-Certificate must be for hash $ \kec{(r''',\eb)}$ as honest validators only include valid Proposal message in their Prepared-Certificates.
    Now consider that since any valid Proposal message for rounds higher than 0 includes \quorum{n_h} Round-Change messages, \Cref{lem:intersection-q-f-with-q} implies that at least one of these contains a Prepared-Certificate sent by an honest-validator, say $v'$, that also sent a valid Commit message for hash $ \kec{(r,\eb)} $ and round $r$. 
    If this is the Prepared-Certificate with the highest round number that is included in the Round-Change-Certificate included in the Proposal message for round $r''$ then the Lemma is proved.
    If not, our consideration above and \cref{lem:no-two-valid-pc-can-be-created} imply that any valid Prepared-Certificate for any round $r'''$, with $r<r'''\leq r'$, is for hash $ \kec{(r''',\eb)}$.
    This and our assumption on the uniqueness property of the Keccak hash function complete the proof.
\end{proof}

\begin{lemma}\label{lem:if-valid-fb-is-produced-then-only-valid-fb-for-the-same-block-can-be-produced}
	If a valid finalised block including the Ethereum block \eb{} is created in round $r$, then any other finalised block created at the same round or later round includes the same Ethereum block \eb. 
\end{lemma}
\begin{proof}
	The proof is by contradiction.
	 Hence, we assume that a finalised block including an Ethereum block $\eb'$, with $\eb\neq \eb'$, is produced at round $r' \geq r$.
	\cref{lem:intersection-of-q-and-} implies that an honest validator, say $v$, sent both a Commit message for hash $ \kec{(r,\eb)} $ and round $r$ and a Commit message for block hash $ \kec{(r',\eb')} $ and round $r'$.
	However, according to \cref{lem:if-finalised-block-is-created-only-that-block-can-be-proposed-afterwards}, only block $\eb$ can be proposed in a valid Proposal message for round $r'$.
	Hence, since honest validators only send Commit messages matching valid Proposal messages, this implies that $v$ can only send Commit messages for block hash  \kec{(r'',\eb)} for any round $r'' \geq r$ which leads to a contradiction.
\end{proof}

\begin{theorem}\label{lem:no-two-differnt-block-can-be-produced}
	The \ibfpmthree{} achieves optimal Byzantine-fault-tolerance threshold.
\end{theorem}
\begin{proof}
	It is an obvious consequence of \cref{lem:if-valid-fb-is-produced-then-only-valid-fb-for-the-same-block-can-be-produced} that no two valid finalised blocks including two different Ethereum blocks can be produced by the \ibfpmthree{}.
\end{proof}

\newcommand{\sih}[1]{\ensuremath{{si}_{h,#1}}}
\newcommand{\cround}[1]{\ensuremath{{cr}_{h,#1}(t)}}
\newcommand{\croundt}[2]{\ensuremath{{cr}_{h,#1}(#2)}}
\newcommand{\sround}[2]{\ensuremath{{sr}_{h,#1}(#2)}}
\newcommand{\snround}[2]{\ensuremath{{nonForcedRoundStart}_{h,#1}(#2)}}
\newcommand{\snroundinv}[2]{\ensuremath{{startNaturalRound}_{h,#1}^{-1}(#2)}}
\newcommand{\enround}[2]{\ensuremath{{endNaturalRound}_{h,#1}(#2)}}
\newcommand{\roundstartedat}[2]{\ensuremath{{roundStartedAt}_{h,#1}(#2)}}
\subsection{Weak-Liveness Analysis of the \ibfpmthree{}}



In this section we show that the \ibfpmthree{} ensures optimal Byzantine-fault-tolerant Weak-Liveness.
We analyse the generic $h$-th instance of the \ibfpmthree{} and, as in the Section above, we assume that, for any instance $h$ of the \ibfpmthree{}, no more than $f(n_h)$ validators are Byzantine.

\begin{lemma}\label{lem:move-to-higher-round-if-round-timer-expires-or-receive-new-round-for-later-round}
	Honest validators running the $h$-th instance of the \ibfpmthree{} move to a round higher than the current one when one of the following events occurs:
	\begin{itemize}
		\item the round timer for the current round expires;
		\item they receive \quorum{n_h} Round-Change messages for a round higher than the current one sent by distinct validators;
		\item they receive a valid Proposal message for a round higher than the current one.
	\end{itemize}
\end{lemma}
\begin{proof}
	Obvious from the pseudocode.
\end{proof}

\begin{lemma}\label{lem:validator-moves-to-higher-round-by-time-r+RoundTimerTimeout}
	If validator $v$ starts round $r$ at time $t$, then $v$ will move to a round $r'>r$ by the time $t+ {RoundTimerTimeout}(r)$.
\end{lemma}
\begin{proof}
	Obvious from the following properties deductible from the pseudocode:
	\begin{itemize}
		\item on the expiry of the round timer for the current round, honest validators move to a higher round;
		\item the round timer for the current round is never restarted;
		\item validator $v$ never moves to a round lower than the current round.
	\end{itemize}
\end{proof}

Let \sround{v}{r} denote the time at which validator $v$ starts round $r$ of instance $h$.
\begin{lemma}\label{lem:first-validator-to-start-the-h-instance-is-the-first-to-start-any-round-in-that-instance}
	Let $v_f$ be the first honest validator of instance $h$ to start instance $h$.
	The following relation is verified for any honest validator $v$ of instance $h$ and round $r$ such that both $v_f$ and $v$ start round $r$ at some point while in instance $h$:
\begin{equation*}
	\sround{v_f}{r} \leq \sround{v}{r}
\end{equation*}

\end{lemma}
\begin{proof}
	The proof is by induction on the round number.
	The Lemma is obviously verified for the base case ($r=0$) as we assume that $v$ starts instance $h$ before any other validator.
	
	For the inductive case, we show by contradiction that if Lemma holds for $r'$ then it also holds for $r = r'+1$.
	Assume that the Lemma holds for $r'$, but not for $r$. 
	Since 
	\begin{enumerate}[label=(\roman*)]
		\item $v_f$ started round $r'$ no later than when $v$ started round $r'$,
		\item the length of the round timer for round $r'$ is identical between $v_f$ and $v$, 
		\item validators move to a higher round either if the round timer for the current round expires, they receive \quorum{n_h} Round-Change messages, sent by distinct validators, for a round higher than the current one or if they receive a valid Proposal message with round number higher than the current one, 
	\end{enumerate}
	the assumption that  $v$ moves to $r$ before $v_f$ does implies that $v$ received either \quorum{n_h} Round-Change messages, sent by distinct validators, for round $r$ or  a valid Proposal message for round $r$ before $v_f$ moves to round $r$.
	This, in turn, implies that at least one honest validator different from $v$, say $v'$, sent a Round-Change message for round $r$, as at least one of the Round-Change messages received or included in the Proposal message has been sent by an honest validator. This implies that $v'$'s round timer for round $r'$ expired before the $v_f$'s round timer for round $r'$.
	This is in clear contradiction with conditions (i) and (ii) above.
\end{proof}

Let \snround{v}{r} denote the time at which round $r$ (and its related round timer) is started if $v$ has moved to a higher round in each round $ < r $ only as effect of the round timer expiry.

\begin{lemma}\label{lem:snround=sround-for-furst-validator}
	Let $v_f$ be the first honest validator of instance $h$ to start instance $h$.
	The following condition is always verified:
\begin{equation*}
	\snround{v_f}{r} = \sround{v_f}{r}
\end{equation*}
\end{lemma}

\begin{proof}
	The proof is by induction.
	The Lemma is obviously verified for the base case $r=0$.
	
	For the inductive case, we show by contradiction that if Lemma holds for $r'$ then it also holds for $r = r'+1$.\\
	\textbf{Case 1:} $\snround{v_f}{r} > \sround{v_f}{r}$. This and \cref{lem:move-to-higher-round-if-round-timer-expires-or-receive-new-round-for-later-round} imply that $v_f$ receives either \quorum{n_h} Round-Change messages, sent by distinct validators, for round $r$ or  a valid Proposal message for round $r$  before the expiry of the round timer for round $r'$.
	For this to happen, as argued in the proof of \Cref{lem:first-validator-to-start-the-h-instance-is-the-first-to-start-any-round-in-that-instance}, there must exist an honest validator different from $v_f$, say $v$, that sends a Round-Change message for round $r$, and therefore starts round $r$ before $v_f$ moves to round $r$.
	This is in contradiction with \cref{lem:first-validator-to-start-the-h-instance-is-the-first-to-start-any-round-in-that-instance}.\\
	\textbf{Case 2: }$\snround{v_f}{r} < \sround{v_f}{r}$. This implies that round $r$ is started after the expiry of round time for round $r'$.
	This is in contradiction with \Cref{lem:validator-moves-to-higher-round-by-time-r+RoundTimerTimeout}.
\end{proof}

\begin{lemma}\label{lem:snround-geq-sround-for-last-validator}
	Let $v_\ell$ be the last honest validator to start instance $h$ of the \ibfpmthree{}.
	The following relation is verified for any honest validator $v$ of instance $h$ and round $r$ such that both $v_\ell$ and $v$ start round $r$ at some point while in instance $h$:
\begin{equation*}
	\snround{v_\ell}{r} \geq \sround{v}{r}
\end{equation*}
\end{lemma}
\begin{proof}
	The proof is by induction.
	The Lemma is obviously verified for the base case $r=0$ as we assume that $v_\ell$ is the last honest validator to start instance $h$, and therefore round 0.
	
	For the inductive case, we show by contradiction that if Lemma holds for $r'$ then it also holds for $r = r'+1$.
	Assume that there exists a validator $v'$ for which $\snround{v_\ell}{r} < \sround{v'}{r}$. 
	Since the Lemma is verified for $r'$, this implies that $v'$ started round $r$ after $\sround{v'}{r'}+\TimeoutForRoundZero{\ensuremath{r'}}$ which contradicts \Cref{lem:validator-moves-to-higher-round-by-time-r+RoundTimerTimeout}.
\end{proof}

Let $ \sih{v} $ be the time at which validator $v$ starts the $h$-th instance of the \ibfpmthree{}.
\begin{lemma}
	The following equation holds for any validator $v$ of the $h$-th instance of the \ibfpmthree{} provided that validator $v$ starts round $r$ at some point while in instance $h$.
	\begin{equation*}
	\snround{v}{r} \equiv \sih{v} + \TimeoutForRoundZero \cdot \left( 2^r - 1\right)
	\end{equation*}
\end{lemma}
\begin{proof}
	\begin{align*}
	\snround{v}{r} &\equiv \sih{v} + \sum_{i=0}^{r-1}\TimeoutForRoundZero \cdot 2^i \\
	&\equiv \sih{v} + \TimeoutForRoundZero \cdot \left( 2^r - 1\right)
	\end{align*}
	
	To be noted that $ \snround{v}{0}$ correctly corresponds to $\sih{v} $
\end{proof}

\newcommand{\minTimeAllHonestValidatorsInTheSameRound}[1]{\ensuremath{{minTimeAllHonestValidatorsAreInTheSameRound}(#1) }}

\begin{lemma}\label{lem:min-t-all-honest-validators-in-the-same-round}
	Let $v_f$ be the first honest validator of instance $h$ to start round $r$, $v_\ell$ be the last honest validator of instance $h$ to start round $r$. $\minTimeAllHonestValidatorsInTheSameRound{r}$, representing the length of the minimum time segment where all honest validators that started instance $h$ are in round $r$ at the same time, is expressed as follows:
	\begin{equation*}
	\minTimeAllHonestValidatorsInTheSameRound{r}\equiv \max( \sih{v_f} - \sih{v_\ell} + \TimeoutForRoundZero \cdot 2^r,0)
	\end{equation*}
\end{lemma}
\begin{proof}

	It is obvious from \Cref{lem:first-validator-to-start-the-h-instance-is-the-first-to-start-any-round-in-that-instance,lem:snround=sround-for-furst-validator,lem:snround-geq-sround-for-last-validator} that $	\minTimeAllHonestValidatorsInTheSameRound{r}\equiv \max(\snround{v_f}{r+1}-\snround{v_\ell}{r},0)$.
	
	The following series of equivalences proves the Lemma:
	\begin{align*}
	\snround{v_f}{r+1} &- \snround{v_\ell}{r} \\
	& \equiv \sih{v_f} \cdot \TimeoutForRoundZero \cdot 2^{r+1} - (\sih{v_\ell} + \TimeoutForRoundZero \cdot 2^r)\\
	& \equiv \sih{v_f} - \sih{v_\ell} + \TimeoutForRoundZero \cdot (2^{r+1}-2^r)\\
	&\equiv \sih{v_f} - \sih{v_\ell} + \TimeoutForRoundZero \cdot 2^r
	\end{align*}	
\end{proof}
To be noted that $ \sih{v_\ell}\geq \sih{v_f} $.

\begin{lemma}\label{lem:q+f-leq-n}
	The following inequality is verified for any $n \geq 0$:
	\begin{equation*}
	\quorum{n} + f(n)\leq n
	\end{equation*}
\end{lemma}
\begin{proof}
	See Lemma 24 of \cite{IBFT1-Analysis}.
\end{proof}

\begin{lemma}\label{lem:conditions-for-termination}
	Let $\Delta$ be the maximum message delay after GST.
	Let $r$ be a round such that:
	\begin{itemize}
		\item at least \quorum{n_h} honest validators start round $r$ after GST;
		\item the proposer of round  $r$ is honest;
		\item $ 	\sih{v_f} - \sih{v_\ell} + \TimeoutForRoundZero \cdot 2^r \geq 4 \cdot \Delta $\\
		where $v_f$ is the first honest validator to start instance $h$ and $v_\ell$ is the last honest validator of a group of  \quorum{n_h} honest validators to start instance $h$.
	\end{itemize}
	\ibfpmthree{} is guaranteed to produce a valid finalised block at round $r$.
\end{lemma}
\begin{proof}
	$\sih{v_f} - \sih{v_\ell} + \TimeoutForRoundZero \cdot 2^r \geq 4 \cdot \Delta$ implies that all of the \quorum{n_h} honest validators are in round $r$ for at least $4\cdot \Delta$ time.
	Let $v_{\ell,r}$ be the latest honest validator, of a group of \quorum{n_h} honest validators, to start round $r$.
	The following sequence of events that leads to the creation of a valid finalised block completes in no more than $4\cdot\Delta$ time.
	\begin{enumerate}
		\item $v_{\ell,r}$ sends a Round-Change message for round $r$ at time $t$. 
		All other validators have already sent a Round-Change message for round $r$.
		\item By time $t+\Delta$ the proposer for round $r$, $p_r$, which is honest by assumption, receives at least \quorum{n_h} Round-Change messages for round $r$.
		$p_r$ then sends a Proposal message for round $r$.
		\item By time $t+2\cdot \Delta$ all honest validators receive the Proposal message and therefore send a Prepare message for the block included in the Proposal message.
		\item By time $t+3\cdot \Delta$ all honest validators receive the $ \quorum{n_h}-1 $ Prepare messages sent by all honest validators (except for the proposer) and therefore send a Commit message matching the Prepare and Proposal messages received.
		\item By time $t+4\cdot \Delta$ all honest validators receive the \quorum{n_h} Commit messages for round $r$ which are sufficient to create a valid finalised block.
	\end{enumerate}
\end{proof}

\begin{lemma}\label{lem:conditions-for-termination-are-eventually-met}
	For any instance $h$ of the \ibfpmthree{} there eventually exists a round $r$ that meets the conditions listed in \Cref{lem:conditions-for-termination}:
	\begin{itemize}
		\item at least \quorum{n_h} honest validators start round $r$ sometime after GST;
		\item the proposer of round  $r$ is honest;
		\item $ 	\sih{v_f} - \sih{v_\ell} + \TimeoutForRoundZero \cdot 2^r \geq 4 \cdot \Delta $\\
		where $v_f$ is the first honest validator to start instance $h$ and $v_\ell$ is the last honest validator of a group of  \quorum{n_h} honest validators to start instance $h$.
	\end{itemize}
\end{lemma}
\begin{proof}
	Let $r_{afterGST}$ be the first round of the $h$-th instance of the \ibfpmthree{} that \quorum{n_h} honest validator start after GST.\\
	Let $synchronousRound(r)$ be the condition $\sih{v_f} - \sih{v_\ell} + \TimeoutForRoundZero \cdot 2^r \geq 4 \cdot \Delta$ with $r \geq r_{afterGST}$.
	
	Since
	\begin{enumerate}[label=(\roman*)]
		\item GST eventually occurs;
		\item $\Delta$ is a constant;
		\item in no longer than ${RoundTimerTimeout}(r_{h,v})$ time from starting the current round $r_{h,v}$, any honest validator $v$ moves to a round number higher than the current one except if it can create a valid finalised block in the current round\label{itm:val-move-up-always};
		\item  if honest validators move to a round, then the new round is higher than the current round\label{itm:val-move-only-up};
		\item  \Cref{lem:q+f-leq-n} guarantees that there exist at least \quorum{n_h} honest validators\label{itm:there-are-q-val}; 
		\item $2^x$ is a strictly monotonically increasing function.
	\end{enumerate} 
	the condition $synchronousRound(r)$  will eventually be met for a round $r$.
	Let $r_{firstSynch}$ be the smaller round number of the $h$-th instance of the \ibfpmthree{} for which condition \linebreak $synchronousRound(r_{firstSynch})$ is true.
	Since the propose function guarantees to select all honest validators of the current \ibfpmthree{} instance for any sequence of \nhmacro[] rounds, there exists a round $r \geq r_{firstSynch}$ where the proposer is honest.
	Let $r_{firstHonestAfterSynch}$ be the first round where the proposer is honest such that $r_{firstHonestAfterSynch} \geq r_{firstSynch}$.
	Since $2^x$ is a strictly monotonically increasing function, if the condition $synchronousRound(r_{firstSynch})$ is true, then the condition $synchronousRound(r_{firstHonestAfterSynch})$ is true as well.
	Also, statements \ref*{itm:val-move-up-always}, \ref*{itm:val-move-only-up} and \ref*{itm:there-are-q-val} above imply that at least \quorum{n_h} honest validators start round $r_{firstHonestAfterSynch}  $ at some point while instance $h$ except if a finalised block for height $h$ is produced for a round lower than $ r_{firstHonestAfterSynch} $.
	This proves that round $ r_{firstHonestAfterSynch} $ meets the conditions listed in \Cref{lem:conditions-for-termination}.
\end{proof}

\begin{theorem}\label{lem:proof-of-weak-liveness}
	The \ibfpmthree{} guarantees optimal Weak-Liveness.
\end{theorem}
\begin{proof}
	Obvious from the definition of optimal Weak-Liveness and \Cref{lem:conditions-for-termination,lem:conditions-for-termination-are-eventually-met}.
\end{proof}

\subsection{Robustness Proof}

\begin{lemma}\label{lem:persistence-proof}
	IBFT 2.0 achieves optimal Persistence.
\end{lemma}
\begin{proof}
	The following considerations prove the Lemma:
	\begin{itemize}
		\item IBFT 2.0 implements the modification IBFT-M1 described in Section 5.1 of \cite{IBFT1-Analysis};
		\item Theorem 1 of \cite{IBFT1-Analysis} proves that modification IBFT-M1 and the guarantee that the \ibfpmthree{} provides optimal Byzantine-fault-tolerant Safety are sufficient conditions to ensure the the IBFT 2.0 protocol guarantees optimal Byzantine-fault-tolerant Persistence;
		\item \cref{lem:no-two-differnt-block-can-be-produced} of this paper proves that \ibfpmthree{} guarantees optimal Byzantine-fault-tolerant Safety.
	\end{itemize}
\end{proof}

The Persistence and Weak-Liveness properties have been proved by relying exclusively on the eventually synchronous network model assumption.
To prove the Liveness property we introduce the following assumption which is a modified version of the Fair Scheduler assumption used by \cite{Bracha:1985:Asynchronous-Consensus-and-Broadcast-Protocols}.
We show in \cref{sec:remove-fairness-network-assumption} how the IBFT 2.0 protocol can be modified to remove the need for this assumption.
\newcommand{\fairnetworkbehaviourassumption}[1]{Fair Network Behaviour Assumption}
\begin{assumption}[\fairnetworkbehaviourassumption{}]
	The \fairnetworkbehaviourassumption{} states that the probability that a message sent by a validator is received by another validator within time $\frac{\TimeoutForRoundZero}{3}$ is higher than 0.

\end{assumption}
\begin{proof}[Justification]
	The following considerations clarify why this assumption is indeed realistic:
	\begin{itemize}
		\item network latency is inherently probabilistic;
		\item it is expected that \TimeoutForRoundZero is set to a value at least 3 times higher that the known lower measured message latency.
	\end{itemize}
\end{proof} 

\begin{lemma}\label{lem:eventually-there-is-a-sequence-of-n-half-blocks-with-no-change-to-the-validator-set}
	Let $\fb_h$ be the height of the finalised block \fb{}.
	For any finalised block \fb{} there eventually exists a finalised block $ \fb' $ with $ {\fb'}_h \geq \fb_h $ such that there is no change to the validator set in the next $ \left \lfloor \frac{n_{{\fb'}_h}}{2} \right \rfloor $ blocks. 
\end{lemma}
\begin{proof}
	Direct consequence of the following invariants of the algorithm for modifying the validator set:
	\begin{itemize}
		\item only one vote for adding or removing validators can be cast per block;
		\item more than half of the validators must cast a consistent vote for the vote to have effect.
	\end{itemize}
	A more detailed proof may be provided in a future version of this work.
\end{proof}

\begin{lemma}\label{lem:if-proposer-for-round-0-is-honest-and-fair-network-behaviour-assumption-is-true-then-block-proposed-is-finalised-in-round-0}
	If the proposer for round 0 of the $h$-th instance of the \ibfpmthree{} is honest and all messages sent by honest validators in round 0 are delivered within time $\frac{\TimeoutForRoundZero}{3}$  then the block proposed by the proposer is finalised within round 0.
\end{lemma}
\begin{proof}
	Let $\delta$ be $\frac{\TimeoutForRoundZero}{3}$ and $p_0$ be the proposer for round 0 of the $h$-th instance of the \ibfpmthree{}.
	The following sequence of events leads to the creation of a valid finalised block completes in no more than $3\cdot\delta$ time.
	\begin{enumerate}
		\item At time $t$, $p_0$ sends a Proposal message for round 0.
		\item By time $t+\delta$ all honest validators receive the Proposal message and therefore send a Prepare message for the block included in the Proposal message.
		\item By time $t+2\cdot \delta$ all honest validators receive the $ \quorum{n_h}-1 $ Prepare messages sent by all honest validators (except for the proposer) and therefore send a Commit message matching the Prepare and Proposal messages received.
		\item By time $t+3\cdot \delta$ all honest validators receive the \quorum{n_h} Commit messages for round $r$ which are sufficient to create a valid finalised block.
	\end{enumerate}
\end{proof}

\begin{lemma}\label{lem:probability-one-of-the-next-f-block-is-prposed-by-an-honest-validator-is-higher-than-0}
	Let \fb{} be a generic valid finalised block with height $h$.
	If the selected proposer logic is Round-Robin (defined at page \hyperlink{def:round-roubin-selection-logic-definition}{ \pageref*{sec:description-of-the-ibfpmthree-protocol}} of \cref*{sec:description-of-the-ibfpmthree-protocol}) and the validator set in the next $f(n_h)+1$ blocks does not change, then the probability that at least one of the next $f(n_h)+1$ finalised blocks contains an Ethereum block created by an honest validator is higher than 0.
\end{lemma}
\begin{proof}
	Assume that a finalised block \fb{} with height $h$ has been created by a validator $v$ while in round $r$.
	Based on the Round-Robin proposer selection logic, the validator $v'$ selected to be the proposer for the first round of the \ibfpmthree{} for the next height, $h+1$, is the same validator that would have been selected as proposer for round $r+1$ for the $h$-th instance of the \ibfpmthree{}.\\
	It is obvious that at least one proposer, say $v''$, out of the next $f(n_h)+1$ validators selected by the proposer selection function is honest.
	The \fairnetworkbehaviourassumption{} implies that the probability that $v''$ receives \quorum{n_h} Round-Change messages and can therefore create a valid Proposal message is higher than 0.
	Also, \Cref{lem:if-proposer-for-round-0-is-honest-and-fair-network-behaviour-assumption-is-true-then-block-proposed-is-finalised-in-round-0} and the \fairnetworkbehaviourassumption{} imply that the probability that the block proposed by $v''$ is finalised is higher than 0.
	This concludes the proof.
\end{proof}

\begin{lemma}\label{lem:a-block-proposed-by-an-honest-validator-will-be-eventually-finalised}
	If the proposer selection logic is Round-Robin, then for any Ethereum block \eb{} included in a valid finalised block, an Ethereum block $\eb'$ with height higher than the height of \eb{} and created by an honest validator will be eventually finalised.
\end{lemma}
\begin{proof}
	Let $\fb_h$ be the height of the generic finalised block \fb{}.
	Based on \cref{lem:eventually-there-is-a-sequence-of-n-half-blocks-with-no-change-to-the-validator-set}, there eventually exists a finalised block $ \fb' $ with $ {\fb'}_h \geq \fb_h $ such that there is no change to the validator set in the next $ \left \lfloor \frac{n_{{\fb'}_h}}{2} \right \rfloor $ blocks. 
	Based on \cref{lem:probability-one-of-the-next-f-block-is-prposed-by-an-honest-validator-is-higher-than-0}, for any block \fb{} such that there is no change to the validator set in the next $ f(n_{\fb_h}) +1$ blocks, the probability that none of the Ethereum blocks included in the next $f(n_{\fb_h}) +1$ blocks is created by an honest validator is lower than 1.
	Since $\left \lfloor \frac{n}{2} \right \rfloor \geq f(n)+1$, this implies that the probability that the Lemma is false approaches 0 as the blockchain height goes to infinity.
	This last implication and the Weak-Liveness property of \ibfpmthree{} (\Cref{lem:proof-of-weak-liveness}) which ensures that eventually a new block is added to the chain prove the Lemma.
\end{proof}

\begin{lemma}\label{lem:liveness-proof}
	The IBFT 2.0 protocol with Round-Robin proposer selection logic guarantees optimal Liveness.
\end{lemma}
\begin{proof}
	If a transaction is sent to all validators, since (i) as proved in \cref{lem:a-block-proposed-by-an-honest-validator-will-be-eventually-finalised}, for any finalised block \fb{} there eventually exists a finalised block that has height higher than \fb{} and that includes an Ethereum block created by an honest validator and (ii) honest validators employee a fair transaction selection logic, the transaction will be eventually included in a finalised block.
\end{proof}

\begin{theorem}
	IBFT 2.0 is a robust blockchain protocol when operating in an eventually synchronous network with no more than $f(n_h)$ Byzantine validators for each blockchain height $h$.
\end{theorem}
\begin{proof}
	Obvious from \cref{lem:persistence-proof,lem:liveness-proof}.
\end{proof}

\section{Improvements}\label{sec:improvements}
In this section we discuss a series of possible improvements applicable to the IBFT 2.0 protocol described in \cref{sec:protocol-description}.
We call this protocol 

\subsection{Reduction of the number of assumptions required for the correctness of the protocol}
\subsubsection{\improvement: Remove the \fairnetworkbehaviourassumption{}}\label{sec:imp:remove-prob-assumption}
\paragraph{Modification}
In order to remove the dependency of the robustness proof on the \fairnetworkbehaviourassumption{} it is sufficient to require that the proposer selection function for height $h$ only selects validators that did not propose any of the latest $f(n_h)$ blocks where $n_h$ is the number of validators for the block with height $h$.

This modification corresponds to replacing the \Proposer{$\cdot, \cdot$} function of the IBFT 2.0 protocol with the $ FairProposer $ function defined below:

\begin{algorithm}[H]
	\setcounter{AlgoLine}{0}
	\caption{$ FairProposer(chain, r) $}
	\Fn{$ FairProposer(chain, r) $}
	{
		$\Let\: n \equiv \N{chain}$\;
		$\Let\: latestfProposers \equiv \text{proposers of the lastest } f(n) \text{ blocks of } chain$\;
		$ i \gets 0 ;$
		$ p \gets \bot$
		\;

		\While{$ i < r+1$}{
			$p \gets \Proposer{chain,i}$\;
			\If{$\Not (p \In latestfProposers) $}
			{
				$i \gets i +1$\;
			}
		}		
		\Return $p$\;
	}
\end{algorithm}

\paragraph{Justification}
This modification reduces the number of possible proposers for the next block from $n_h$ down to $n_h - f(n_h)$.
Which implies that the minimum number of honest proposers for the next block is reduced from $n_h - f(n_h)$ down to $n_h - 2 \cdot f(n_h)$.
It is quite easy to prove that $n_h - 2 \cdot f(n_h) > 0$ which means that this modification guarantees that at least one of the possible proposers for the next block is honest.
It is quite easy to see how this in turn implies that this modification does not affect the validity of the Weak-Liveness and Liveness properties proved in \cref{sec:robustness-analysis}.

For any finalised block \fb{} with height $h$, provided that there is no change to the validator list in the next $f(n_h)+1$ blocks, since we assume that no more than $f(n_h)$ validators are Byzantine, this modification ensures that at least one out of the next $f(n_h)+1$ finalised blocks includes an Ethereum block created by an honest validator.
It should be noted how this last statement can be used as sketch proof for \cref{lem:a-block-proposed-by-an-honest-validator-will-be-eventually-finalised} without depending on the \fairnetworkbehaviourassumption{} any more.

\subsection{Reduction of the block finalisation latency}
\subsubsection{\improvement: Reduce the minimum number of communication phases from three down to two}\label{sec:remove-fairness-network-assumption}
\paragraph{Modification}
In the IBFT 2.0 protocol, the minimum number of communication phases required to finalise a block is three even if all validators are honest and the message transmission latency is less than $\frac{\TimeoutForRoundZero}{3}$:
\begin{enumerate}
	\item The proposer for round 0 and instance $h$ of the \ibfpmthree{} sends a Proposal message to all validators;
	\item All validators of height $h$, excluding the proposer for round 0, reply by sending a Prepare message to all validators;
	\item All validators of height $h$, once they receive the Proposal message for round 0 from the proposer for round 0 and $ \quorum{n_h} -1 $ Prepare messages from non-proposing validators for round 0, they send a Commit message to all validators.
	\item[] Once a validator for height $h$ receives \quorum{n_h} Commit messages for round 0 from the validators of height $h$, it creates a finalised block;
\end{enumerate}

The modification described in this section, which is based on the  very fast learning protocol presented by \citewithauthor{best-case-complexity-asynch-byz-consnesus}, allows reducing the minimum number of communication phases from three down to two:
\begin{enumerate}
	\item The proposer for round 0 and instance $h$ of the \ibfpmthree{} sends a Proposal message to all validators;
	\item All validators of height $h$, excluding the proposer for round 0, reply by sending a Prepare message to all validators;
	\item[] Once a validator for height $h$ receives the Proposal message for round 0 from the proposer for round 0 and $ n_h -1 $ Prepare messages from all the non-proposing validators for round 0 it creates a finalised block.
\end{enumerate}
Block finalisation in only two communication phases can be accomplished only if all the following conditions are met:
\begin{itemize}
	\item  all validators are honest; 
	\item the network latency is less than $\frac{\TimeoutForRoundZero}{2}$.
\end{itemize}
If we assume the same latency for the different communication phases, then when the conditions listed above are met, this modification reduces the block finalisation latency by 33\%.
To be noted that when the conditions listed above are not achieved, then the protocol degrades back to the performance of the original IBFT 2.0 protocol with no overhead.

This improvement requires applying the following modifications to the IBFT 2.0 protocol:
\begin{itemize}
	\item If \cref{imp:remove-seals} is not applied, then a ``Prepare seal'' mast be added to the Prepare messages.
	If \cref{imp:remove-seals} is applied, then the Prepare messages do not need to be modified;
	\item If a block is finalised in only two phases, then the commit seals include the signatures (or Prepare seals, if \cref{imp:remove-seals} is not applied) of the $n_h -1 $ Prepare messages;
	\item During the validation of a finalised block, if the number of commit seals is $n_h -1$ then the seals are interpreted as signature (Prepare seals, if \cref{imp:remove-seals} is not applied) of Prepare messages, otherwise if the number of seals is \quorum{n_h} then they are interpreted as signature (Commit seals, if \cref{imp:remove-seals} is not applied) of Commit messages;
	\item When a validator sends a Round-Change message, if the validator has received a valid Proposal message for round 0, but it has never \nprepared{}, then it must include the Proposal message in the Round-Change message.
	After the first time that a validator \nprepare{s} while running instance $h$ of the \ibfpmthree{}, the validator will only include the latest Prepared-Certificate in any Round-Change message that it sends;
	\item When a validator sends a Proposal message for a round number higher than 0, if there exists an Ethereum block \eb{} such that: 
	\begin{itemize}
		\item \eb{} is the only one Ethereum block such that at least $f(n_h)+1$ of the Round-Change messages forming the Round-Change-Certificate (included in the Proposal message) include a Proposal message for round 0 and for this block;
		\item no Round-Change message include a valid Prepared-Certificate
	\end{itemize}
	then the proposed block included in the Proposal message must include the Ethereum block \eb{}.\\
	If the conditions listed above are not met, then the Ethereum block to be included in a Proposal message with height higher than 0 must be determined as specified by the original IBFT 2.0 protocol;
	\item On the reception of a Proposal message for a round higher than 0, the same calculation must be performed to validate the Proposal message.
\end{itemize}

As per the description above, this modification allows reducing the minimum number of communication phases from three down to two only for round 0.
While it is possible to modify the IBFT 2.0 protocol to reduce the minimum number of communication phases down to two for any round, because of the following reasons, we believe that this optimisation should be applied only to round 0:
\begin{itemize}
	\item applying the optimisation for rounds higher than 0 as well requires potentially increasing the size of the Round-Change message by the size of one block as the latest accepted Proposal message and the latest Prepared-Certificate can refer to two different blocks;
	\item if the conditions required to achieve block finalisation in two phases are not met for round 0, then either (i) one or more validators are Byzantine and can therefore prevent block finalisation in two phases at any round or (ii) the network latency is higher than $\frac{\TimeoutForRoundZero}{2}$ which means that block finalisation latency is already higher than \TimeoutForRoundZero and therefore whether block finalisation is reached in two or three phases at the next round adds only minimal improvement to the block finalisation latency.
\end{itemize} 

\paragraph{Justification}
In this paragraphs we provide a sketch proof to show that this modification does preserve the Safety property of the \ibfpmthree{}.\\
Since any valid Proposal message for round higher than 0 includes \quorum{n_h} Round-Change messages,  if the Ethereum block \eb{} is finalised in two phases, then at least $ \quorum{n_h} - f(n_h) $ of the Round-Change messages included in the Round-Change-Certificate included in the Proposal message for round 1 have been sent by honest validators and they will therefore include the Proposal message received by the proposer of round 0.
Since $ \quorum{n_h} - f(n_h) \geq f(n_h)+1$ (quite easy to prove), the number of Proposal messages for the Ethereum block \eb{} is $\geq f(n_h)+1$.
Also, since all honest validators will include the Proposal message for the Ethereum block \eb{} in the Round-Change that they send for round 1 and there are at most $f(n_h)$ Byzantine validators, then it is impossible that there exists another Ethereum block $\eb'$ different from \eb{} for which $f(n_h)+1$ Round-Change messages including matching Proposal messages are sent.
These two considerations imply that the only Ethereum block that a valid Proposal message for round 1 can include is \eb.
Because honest validators can only \nprepare{} on the Ethereum block included in a valid Proposal message, the only block that any honest validator can \nprepare{} on in round 1 is \eb.
Therefore, any Round-Change for round 2 sent by honest validators will either include a Proposal message for the Ethereum block \eb, or a Prepared-Certificate for the same Ethereum block \eb.
This argument can be generalised to any round by induction on the round number.

\subsubsection{\improvement: Reduce the latency required for \quorum{n_h} validators to be in the same round for sufficient time to create a finalised block}
\paragraph{Modification}
This improvement requires applying the following modifications to the IBFT 2.0 protocol.
When a validator running the $h$-th instance of the \ibfpmthree{} receives $f(n_h)+1$ Round-Change messages for height $h$ and round number higher than the current round number, it (i) starts the round number corresponding the lowest round number amongst the set of $f(n_h)+1$ Round-Change messages received that have round number higher than the current one  and (ii) sends a Round-Change messages for this round number.

\paragraph{Justification}
Since there are no more than $f(n_h)$ Byzantine validators, if a validator $v$ receives $f(n_h)+1$ Round-Change messages for height $h$ and round number higher than the current round number, then at least one honest validator is in a round higher than the current one.
Specifically, at least one honest validator is in a round equal to or higher than the lowest round number amongst the set of $f(n_h)+1$ Round-Change messages received.
Therefore moving to this round allows validator $v$ to start this higher round sooner than in the original IBFT 2.0 protocol.
It can be seen how this modification does not affect any of the Lemmas used to prove the robustness of the IBFT 2.0 protocol.

\subsection{Decrease message size}
\subsubsection{\improvement: Replace the Round-Change messages included in the Round-Change-Certificate included in Proposal messages with the hashes of the Round-Change messages}
\paragraph{Modification}
This improvement requires replacing the Round-Change messages included in the Round-Change-Certificate included in Proposal messages with the hashes of the Round-Change messages.
We call the set of hashes of Round-Change messages included in a Proposal message as effect of this modification Round-Change-Hash-Certificate.
This improvement requires validators to cache Round-Change messages that they receive and then verify a Proposal message by using a Round-Change-Certificate composed of the received Round-Change messages for which Keccak hash matches one of the hashes included in the Round-Change-Hash-Certificate.

If this modification is applied, then the block hash field can be removed from the Proposal message and the proposed block field can be moved inside the signed portion of the message.
The reason why Proposal messages in the original IBFT 2.0 protocol include the block hash in the signed portion of the message is to reduce the size of Proposal messages with round higher than 0.
This because in the original IBFT 2.0 protocol a signed Proposal message with enough information to ``bind'' it to a proposed block must be included in Proposal messages for round higher then 0 as part of the Prepared-Certificates included in the Round-Change-Certificate.
It should be obvious why if this improvement is applied, then the size of Proposal messages for rounds higher than 0 becomes independent of the block size.

Similarly, if this improvement is applied then the size of Proposal messages for rounds higher than 0 becomes independent of the Round-Change message size.
Therefore, there is no more need to split the Round-Change message between a signed portion and an unsigned portion and all message fields can be moved within the signed portion of the message.

\paragraph{Justification}
The justification for this modification is based on how the Gossip protocol works.
Specifically, when a validator receives a message from one of its peers, it transmits that message to all other peers.
This is essentially similar to how the transmission of a message works: a message is sent to all peers.
The only difference between the reception and the transmission is that when a message is received the message is not transmitted to the peer that sent it as that peer already has the message.
Hence, in the IBFT 2.0 protocol, when a validator sends a Proposal message for a round higher than 0, the Round-Change messages included in the Round-Change-Certificate have already been transmitted to the same peers that the Proposal message will be transmitted to.
Since a validator sends a Proposal message as soon as it receives \quorum{n_h} Round-Change messages, the Round-Change messages included in the Proposal message are transmitted to the peers not long before the Proposal message is transmitted as well.\\
To be noted that the Round-Change-Certificate cannot be removed from the Proposal message as validators receiving a Proposal message must be able to determine whether any valid block  or only a \nprepared{} block could have been included as proposed block in the Proposal message.
If the Round-Change-Certificate is completely removed then the only way for validators to validate Proposal messages is to receive a Round-Change message from each of the $n_h$ validators which could impair the liveness of the protocol as $f(n_h)$ of the validators may be Byzantine and never send a Round-Change message.


\subsubsection{\improvement:\label{imp:remove-seals} Remove the commit seal from the Commit message and replace the commit seals included in a block with the signatures of the Commit messages received}
\paragraph{Modification}
This improvement requires the following modifications to the IBFT 2.0 protocol:
\begin{itemize}
	\item remove the commit seal from the Commit message;
	\item when composing a finalisation proof, collate the signatures of \quorum{n_h} valid Commit messages sent by distinct validators;
	\item when validating a finalised block, reconstruct the body of the Commit message that would have been sent for the block under validation and verify that
	 each of the signatures included in the finalisation proof is a valid signature of the reconstructed Commit message by one of the validators. 
	This can be done via the Elliptic Curve Signature Recovery function.
\end{itemize}

\paragraph{Justification}
Finalised blocks include all of the information required to reconstruct the body of the Commit message sent for that specific block, namely: the block height, the round number at which the \quorum{n_h} Commit messages where received and the hash of the block.

\subsubsection{\improvement: Reduce the number of messages included in  Prepared-Certificates and the number of commit seals included in finalisation proofs to the minimum required to guarantee the robustness of the protocol}
\paragraph{Modification}
This improvement requires to apply the following modifications to the IBFT 2.0 protocol:
\begin{itemize}
	\item for any Prepared-Certificate included in a Round-Change message, include only $ \quorum{n_h} - 1 $ Prepare messages for the current round, instance of the \ibfpmthree{} and block hash matching the accepted Proposal, even if more Prepare messages with the same characteristics have been received;
	\item for any finalisation proof, include only \quorum{n_h} commit seals included in valid Commit messages for the current round, instance of the \ibfpmthree{} and block hash matching the accepted Proposal, even if more commit seals included in valid Commit messages have been received.
\end{itemize}

\paragraph{Justification}
As per IBFT 2.0 protocol definition, all of the received Prepare messages for the current round, instance of the \ibfpmthree{} and block hash matching the accepted Proposal message are included in a Prepared-Certificate.
Similarly, as per IBFT 2.0 protocol definition, all of the available commit seals received as part of Commit messages for the current round, instance of the \ibfpmthree{} and block hash matching the accepted Proposal are included in a finalisation proof.
As it can be seen from the robustness analysis in \cref{sec:robustness-analysis}, the number of messages included in Prepared-Certificates and the number of commit seals included in finalisation proofs can be reduced as indicated by this improvement without affecting the robustness of the protocol.

\subsubsection{\improvement: Remove the Proposal message from the Prepared-Certificate}
\paragraph{Modification}
The Proposal message included in the Prepared-Certificates can be removed without impacting the robustness of the protocol.
While this improvement reduces the message size only marginally, it allows simplifying a bit the Prepared-Certificate validation logic.

\paragraph{Justification}
In this paragraph we provide a sketch proof to show how this modification does not affect the robustness of the protocol.
Specifically, we show that \cref{lem:no-two-valid-pc-can-be-created}, which states that	``\emph{no two Prepared-Certificates for the same round and for different block hashes can be created}'',  still holds if this improvement is applied.

\begin{proof}
	By contradiction assume that the Lemma is false and two Prepared-Certificates, say $PC$ and $PC'$ are created for the same round $r$ but for different block hashes, say $H$ and $H'$ respectively.
	We distinguish two cases: the proposer for round $r$ is honest and the proposer for round $r$ is Byzantine.\\
	\textbf{Case 1: the proposer for round $r$ is honest.} 
	In each Prepared-Certificate there are at least \linebreak $\quorum{n_h} - f(n_h) -1$ Prepare messages sent by honest validators.
	It is easy to verify that this number is always higher than 0.
	Since honest validators send only Prepare messages matching the accepted Proposal message, the honest validators that sent the Prepare messages included in $ PC $ must have received a Proposal message for block hash $ H $, whereas the honest validators that sent the Prepare messages included in $ PC' $ must have received a Proposal message for block hash $ H' $.
	This is in clear contradiction with the invariant (easily deductible from the pseudocode) that honest proposers never send different Proposal messages for the same round number.\\
	\textbf{Case 2: the proposer for round $r$ is Byzantine.} 
	Since the Prepare messages included in Prepared-Certificates are sent by non-proposing validators, no more than $ f(n_h)-1 $ Prepare messages in any Prepared-Certificate are sent by Byzantine validators.
	Conversely, at least $ \quorum{n_h} - 1 - (f(n_h)- 1) \equiv \quorum{n_h} - f(n_h)$ of the Prepare messages are sent by honest validators.
	Since any set of \quorum{n_h} validators contains at least $ \quorum{n_h} - f(n_h) $ honest validators, \Cref{lem:intersection-of-q-and-} implies that the intersection of any two sets of $ \quorum{n_h} -f(n_h) $ honest validators is not empty and therefore includes at least one honest validator.
	From here the proof proceeds like in \cref{lem:no-two-valid-pc-can-be-created}.	
\end{proof}

\subsubsection{\improvement:\label{imp:remove-repeated-message-bodies} Remove repeated message bodies from the Prepared-Certificate included in a Round-Change message}
\paragraph{Modification}
Replace the set of Proposal and Prepare messages included in a Prepared-Certificate with a set including only the following pieces of information:
\begin{itemize}
	\item round number of the messages included in the Prepared-Certificate;
	\item block hash of the block included in the Proposal and Prepare messages included in the Prepared-Certificate;
	\item the signature of the Proposal message included in the Prepared-Certificate;
	\item all the signatures of the Prepare messages included in the Prepared-Certificate.
\end{itemize}
A Prepared-Certificate constructed as described above can be validated by reconstructing the body of the Proposal and Prepare messages included in the original Prepared-Certificate by using the pieces of information listed above (the height is provided as part of the Round-Change message) and verifying that (i) the signature for the Proposal message is a valid signature of the proposer (for the round number included in the modified version of the Prepared-Certificate)  over the Proposal message body and (ii) each Prepare signature is a correct signature of one of the non-proposing validators (for the round number included in the modified version of the Prepared-Certificate)  over the Prepare message body.

\paragraph{Justification}
The body of all the Prepare messages included in a Prepared-Certificate is identical.
The body of a Proposal message varies from the body of a Prepare message only in the message type.

\subsubsection{\improvement: Use an aggregate signature scheme for the Prepared-Certificate}
\paragraph{Modification}
The size of the Prepared-Certificate can be further reduced by employing an aggregate signature scheme which allows replacing the $ \quorum{n_h}-1 $ signatures of the Prepare messages with a single signature.

\subsection{Improve the semantics of the messages }
\subsubsection{\improvement: Do not consider the round number when computing the block hash included in IBFT 2.0 messages}
\paragraph{Modification}
As per definition of the IBFT 2.0 protocol, the block hash included in the IBFT 2.0 messages is calculated over a tuple composed of the Ethereum block and the current round number.
This improvement requires replacing this hash with the hash of the Ethereum block only.
To be noted that it is still required to include the current round number in the finalisation proof.

\paragraph{Justification}
The purpose of the $h$-th instance of the \ibfpmthree{} is to decide on the Ethereum block to be added at height $h$.
Therefore, using the PBFT \cite{Castro:1999:PBF:296806.296824} terminology, the \emph{value} to be decided by each instance of the \ibfpmthree{} is the Ethereum block, not the tuple composed by the Ethereum block and the current round number.
The round number is added to the finalisation proof exclusively because the finalisation proof composed of \quorum{n_h} commit seals is valid only if these commit seals were included in Commit messages targeting the same round number.


%
%
%

\section*{Acknowledgements}
Thanks to Robert Coote(PegaSys) who reviewed this work and provided insightful comments on how to improve its clarity and readability.
The author also wishes to thank Shahan Khatchadourian and Ben Edgington who reviewed some early pieces of work that were fundamental in solidifying the thinking required for this paper.
This work has been carried out as part of the author's research activity at PegaSys (ConsenSys).

\printbibliography

\end{document}